\documentclass[11pt,draft]{amsart}
\usepackage{amsfonts,amscd,amsthm,amsgen,amsmath,amssymb, hyperref}
\usepackage[all]{xy}
\usepackage[vcentermath]{youngtab}
\usepackage[letterpaper,hmargin=1.55in,vmargin=1.2in]{geometry}

\usepackage{graphicx, xcolor}
\usepackage{bm}	
\usepackage{amsmath}
\usepackage{amsthm}	
\usepackage{amsfonts}	
\usepackage{amssymb}
\usepackage{fixme}
\usepackage{fixltx2e}
\usepackage{mathrsfs}
\usepackage{latexsym}
\usepackage{enumerate}
\usepackage{dsfont}
\usepackage{amssymb}
\usepackage{amsmath,amssymb,amsfonts,amsthm} 
\usepackage{graphics,graphicx}                 
\usepackage{color}                    
\usepackage{hyperref,fancyhdr}                 
\usepackage[all]{xy}
\usepackage[applemac]{inputenc}
\usepackage[ngerman,english]{babel}
\usepackage{url}
\usepackage{abstract}

\theoremstyle{plain}
\newtheorem{theorem}{Theorem}[section]
\newtheorem{corollary}[theorem]{Corollary}
\newtheorem{lemma}[theorem]{Lemma}
\newtheorem{proposition}[theorem]{Proposition}

\theoremstyle{definition}
\newtheorem{definition}[theorem]{Definition}
\newtheorem{remark}[theorem]{Remark}
\newtheorem{example}[theorem]{Example}

\theoremstyle{remark}

\numberwithin{equation}{section}
\numberwithin{figure}{section}

\newcommand{\Aut}{\text{Aut}}
\DeclareMathOperator{\Hom}{Hom}

\newcommand\id{\mathrm{id}}

\newcommand{\Ext}{\text{Ext}}
\newcommand{\im}{\text{im}}
\newcommand{\Lin}{\text{Lin}}
\newcommand{\pr}{\text{pr}}

\begin{document}

\title{Third group cohomology and gerbes over Lie groups}

\author{Jouko Mickelsson and Stefan Wagner}
\address{Department of Mathematics, University of Helsinki; Department of Mathematics, University of Hamburg}

\email{jouko.mickelsson@gmail.com, stefan.wagner@uni-hamburg.de}

\maketitle

\begin{abstract} 
The topological classification of gerbes, as principal bundles with the structure group
the projective unitary group of a complex Hilbert space,  over a topological space $H$ is given by the third cohomology $\text{H}^3(H, \Bbb Z).$ When $H$ is a topological group the integral cohomology is often related to a locally continuous (or in the case of a Lie group, locally smooth) third group cohomology of $H.$ We shall study in more detail this relation in the case of a group extension $1\to N \to G \to H
\to 1$ when the gerbe is defined by an abelian extension $1\to A \to \hat N \to N \to 1$ of $N.$ 
In particular, when $\text{H}_s^1(N,A)$ vanishes we shall construct a transgression map $\text{H}^2_s(N, A) \to \text{H}^3_s(H, A^N)$, 
where $A^N$ is the subgroup of $N$-invariants in $A$ and the subscript $s$ denotes the locally
smooth cohomology. Examples of this relation appear in gauge theory which are discussed in the paper. 
\end{abstract}

\vspace*{0,3cm}

	

\section{Introduction}

A gerbe over a topological space $X$ is a principal bundle with fiber $PU(\mathcal H),$ the projective unitary
group $U(\mathcal H)/S^1$ of a complex Hilbert space $\mathcal H.$ Gerbes are classified by the integral cohomology
$\mathrm{H}^3(X, \Bbb Z).$ In real life gerbes most often come in the following way. We have a vector bundle
$E$ over a space $Y$ with model fiber $\mathcal H$ with a free right group action $Y \times N \to Y$ with $Y/N = X$ such that the
action of $N$ can be lifted to a projective action of $N$ on $E.$ Then $E/N$ becomes a projective
vector bundle over $X$ with structure group $PU(\mathcal H).$ Put in the groupoid language, we have an $S^1$-extension
of the transformation groupoid $Y\times N \to Y;$ the circle extension can be written in terms of a 2-cocycle
$c_2(y; g, g') \in S^1$ with $y\in Y$ and $g,g' \in N.$ In general, there might be topological obstructions
to define $c_2$ as a globally continuous cocycle on $N.$ A well known example of this is when $N =LK,$ the smooth loop
group of a compact simple Lie group $K.$ This group has nontrivial central $S^1$-extensions parametrized by
$\mathrm{H}^2(N, \Bbb Z)$ but because these are nontrivial circle bundles over $LK$ the 2-cocycle $c_2$ is only
locally smooth, in a neighborhood of the neutral element.

More complicated examples arise in gauge theory. There $N$ is the group of smooth gauge transformations
acting on the space $Y$ of  gauge connections and in the quantum theory one is forced to deal with extensions
of $N$ with the abelian group $A$ of circle valued functions on $Y.$ 

As pointed out by Neeb \cite{Ne04} the abelian extensions, even in the topologically nontrivial case 
$1 \to A \to \hat N \to N \to 1$ of an $A$-bundle over $N,$ can be constructed in terms of locally smooth
2-cocycles.  The natural question is then to which extent gerbes over a Lie group $H$ are described by
locally smooth 3-cocycles on $H.$ We shall concentrate on the following situation. Let $G$ be an (infinite-dimensional) Lie group, $N$ a normal subgroup of $G$,  $H= G/N$ and $A$ a $G$-module. Given an extension
$\widehat{N}$ of $N$ by $A$ (viewed as an abelian group) described by a representative of an element $c_2\in \mathrm{H}^2_s(N, A),$ i.e.,
by a locally smooth 2-cocycle with values in $A,$ we want to define a transgression $\tau(c_2) \in \mathrm{H}^3_s(H,A^N),$
where $A^N$ is the set of $N$ invariant elements in $A.$ This we achieve by assuming that the first cohomology group $\mathrm{H}_s^1(N, A)$
vanishes (cf. Theorem \ref{Gammasmoothcrossedmodule}). In addition, under the same restriction, we show that the transgressed element
$\tau(c_2)$ vanishes if and only if there is a prolongation of the group extension $1\to N \to G \to H \to 1$
to an extension $1 \to \widehat{N} \to \widehat{G} \to H \to 1$ (cf. Theorem \ref{extension}). 

A gerbe over $H$ can be always trivialized in the pull-back on $G$ if $\mathrm{H}^3(G, \Bbb Z) =0.$ In particular,
we can take $G$ as the path group $$PH= \{f: [0,1]\to H:\, f(0) = 1\},$$ which is contractible and is a model
for $E\Omega H,$ the total space of the universal classifying bundle for the based loop group $\Omega H$ 
with $B \Omega H = H.$ Here we could work either in the smooth or continuous category of maps.
In the smooth setting this case is actually relevant in gauge theory applications which we shall discuss
in the last section. For example, when $H$ is a compact Lie group, then the smooth version of $PH$ can be identified with the space of gauge connections of an $H$-bundle over the unit circle through the map $f\mapsto f^{-1}df$ and $\Omega H$ (assuming periodic boundary condititions for $f^{-1}df$) with
the group of based gauge transformation. In this case $H$ is the moduli space of gauge connections. In the same way, if
$N=Map_0(S^n, K)$ is the group of based gauge transformation of a (trivial) $K$-bundle over $S^n$, then
the moduli space of gauge connections on $S^n$ is, up to a homotopy, equal to the group $H = Map_0(S^{n-1}, K).$ These cases become 
interesting when $n$ is odd since there are nontrivial abelian extensions of $N$ described by an elements in
$\mathrm{H}^2_s(N, A)$, where $A$ is the abelian group of circle valued functions on $G= PH.$ In quantum field
theory elements of $\mathrm{H}^2_s$ describe Hamiltonian anomalies, due to chiral symmetry breaking 
in the quantization of Weyl fermions \cite{Mi3}, Chapters 4 and 5, and \cite{CMM} for an index theory derivation.

The 3-cocycles on $H$ are also relevant in the categorical representation theory. 
Consider a category $\mathcal C$ of irreducible 
representations of some associative algebra $\mathcal B$ in a complex vector space $V,$ the morphisms being 
isomorphisms  of representations. Suppose that for each $h\in H$
there is a given functor $F_h : \mathcal C \to  \mathcal C$ and for each pair $h,k\in H$ an isomorphism
$i_{h,k} : F_h \circ F_k \to F_{hk}.$ Then both $i_{j, hk} \circ i_{h,k}$ and $i_{jh,k}\circ i_{j,h}$
are isomorphisms $F_j\circ F_h \circ F_k \to F_{jhk}$  but not necessarily equal; 
$$ i_{j,hk}\circ i_{h,k} = \alpha(j,h,k) i_{jh,k}\circ i_{j,h}$$
for some $\alpha(j,h,k)$ in $\Bbb C^{\times}.$ The function $\alpha$ is then a $\Bbb C^{\times}$-valued cocycle
on $H.$ 

A simple example of the above functorial representation is the case studied in \cite{Mi2}. There the algebra 
$\mathcal B$ is the algebra of canonical anticommutation relations (CAR) generated by vectors  in the space
$V=L^2(S^1, \Bbb C^n)$ of square-integrable vector valued functions on the unit circle and the category of
representations are special quasi-free representations of $\mathcal B$ which are parametrized by
polarizations $V = V_+(a) \oplus V_-(a),$ to positive/negative  modes of the Dirac operator on the circle coupled to
the gauge connection $a= f^{-1}df$ for $f \in PH.$ (One must choose boundary conditions for $f:[0, 2\pi] \to
H$ such that $a$ is smooth and periodic.) 
The group $H$ acts by $n\times n$ unitary matrices on $\Bbb C^n$ and
the group $G= PH$ acts by point-wise multiplication on functions.  Fix a (locally smooth) section $\psi: H
\to G.$ Then each $h\in H$ defines an automorphism of the algebra $\mathcal B$ through the action of
$\psi(h)$ in $V.$ The quasi-free representations of $\mathcal B$ are defined in a Fock space $\mathcal F.$
Given a pair $h,k\in H$ the element $\ell_{h,k} := \psi(h)\psi(k)\psi(hk)^{-1}$ is contained in $N=\Omega G \subset G$
and choosing an element $i_{h,k} \in \hat N$ in the fiber above $\ell_{h,k}$ one obtains an automorphism
of a quasi-free representation parametrized by $a.$ The functor $F_h$ for $h\in H$ takes the the quasi-free
representation corresponding to any $\psi(k)$ to the quasi-free representation parametrized by $\psi(kh).$ The role of Lie algebra 3-cocycles in $C^*$ algebraic constructions in quantum field theory was also studied earlier in \cite{CGRS}.

The article is structured as follows. After this introduction and a preliminary section on infinite-dimensional Lie theory, we provide in Section 3 conditions which assure that certain automorphisms defined on the building blocks of an abelian Lie group extension can be lifted to the extension itself. Next, let $G$ be an (infinite-dimensional) Lie group, $N$ a normal subgroup of $G$,  $H= G/N$ and $A$ a $G$-module. Moreover, let $\widehat{N}$ be an extension of $N$ by $A$ (viewed as an abelian group) described by a representative of an element $c_2\in \mathrm{H}^2_s(N, A),$ i.e., by a locally smooth 2-cocycle with values in $A$. The aim of Section 4 is to find a transgression $\tau(c_2) \in \mathrm{H}^3_s(H,A^N),$ where $A^N$ is the set of $N$ invariant elements in $A.$ This we achieve by using methods from the theory of smooth crossed modules and under the assumption that the cohomology group $\mathrm{H}_s^1(N, A)$ vanishes (cf. Theorem \ref{Gammasmoothcrossedmodule}). In addition, under the same restriction, we show that the transgressed element $\tau(c_2)$ vanishes if and only if there is a prolongation of the group extension $1\to N \to G \to H \to 1$ to an extension $1 \to \widehat{N} \to \widehat{G} \to H \to 1$ (cf. Theorem \ref{extension}). 
In Section 5 we shall give a different construction of lifting the automorphisms of $N$ 
to automorphisms of an abelian extension $\widehat{N}$ in the case when $A$ is the $G$-module of smooth maps from $G$ to the unit circle $S^1.$  This case is important in the gauge theory applications. Section 6 presents some results on the transgression $\mathrm{H_s^2}(N, A) \to \mathrm{H_s^3}(H, A^N)$ in the case of the module $A= Map(G, S^1)$, while Section 6 is devoted to the case of a non simply connected base $H$, e.g., a torus and, further, to a gauge theory application. The goal of the first part of the Appendix is to extend the geometric construction of the central extension $\widehat{LG}$ of the
smooth loop group $LG,$ with $G$ a simply connected compact Lie group, \cite{Mi1}, to the case when $G$ is a connected compact Lie group but not necessarily simply connected. In the second part of the Appendix we discuss a Fock space construction of gerbes over $G/Z$. Finally, in the third part of the Appendix we show that a certain second cohomology group acts in a natural way as a simply transitive transformation group on the set of equivalence classes of extensions of a given smooth crossed module, which is related to our results in Section 4.

We would like to express our greatest gratitude to the referee for providing very fruitful criticism. This work was partially supported by the Academy of Finland grant nr 1138810.

\section{Prelimenaries and notations}\label{pre and not}

In our study we use many tools from infinite-dimensional Lie theory. As a consequence, we provide in this preliminary section the most important definitions and notations which are repeatedly used in this article. For a detailed discussion of this infinite-dimensional context we refer to \cite{Mil83,Ne06}. 


Throughout this paper a \emph{Lie group} $G$ is a smooth manifold modelled on a locally convex space for which the group multiplication and the inversion are smooth maps. We denote by $1_G\in G$ the identity of $G$ and by $C_G:G\rightarrow\Aut(G)$ the conjugation action of $G$. For a normal subgroup $N$ of $G$, we write $C_{N}:G\rightarrow\Aut(N)$ with $C_{N}(g):=C_{G}(g)_{\mid N}$. 

A subgroup $H$ of a Lie group $G$ is called a $\emph{split Lie subgroup}$ if it carries a Lie group structure for which the inclusion map $i_H:H\hookrightarrow G$ is a morphism of Lie groups and the right action of $H$ on $G$ defined by restricting the multiplication map of $G$ to a map $G\times H\rightarrow G$ defines a smooth $H$-principal bundle. This means that the coset space $G/H$ is a smooth manifold and that the quotient map $\pr:G\rightarrow G/H$ has smooth local sections.

An $\emph{extension of Lie groups}$ is a short exact sequence
\[1\longrightarrow N\stackrel{i}{\longrightarrow}\widehat{G}\stackrel{q}{\longrightarrow}G\longrightarrow1
\]of Lie group morphisms, for which $N\cong \ker(q)$ is a split Lie subgroup. This means that $\widehat{G}$ is a smooth $N$-principal bundle over $G$ and $G\cong\widehat{G}/N$. 


In general, we cannot expect the group $\Aut(G)$ of automorphism of a Lie group $G$ to carry a Lie group structure, but any ``reasonable" Lie group structure on this group should have the property that for any smooth manifold $M$, a map $f:M\rightarrow\Aut(G)$ is smooth if and only if the induced maps 
\begin{align*}
f^{\wedge}:M \times G\rightarrow G, \quad f^{\wedge}(m,g):=f(m)(g) 
\end{align*}
and
\begin{align*}
(f^{-1})^{\wedge}:M \times G\rightarrow G, \quad (f^{-1})^{\wedge}(m,g):=f(m)^{-1}(g)
\end{align*}
are smooth. We therefore turn the later condition into a definition and call a map $f:M\rightarrow\Aut(G)$ smooth if the induced maps $f^{\wedge}$ and $(f^{-1})^{\wedge}$ are smooth. We point out that if $M$ is a Lie group and $f$ a group homomorphism, then it is easily seen that the map $f^{\wedge}$ is smooth if and only if the map $(f^{-1})^{\wedge}$ is smooth.


Let $G$ and $N$ be Lie groups. For $p\in\mathbb{N}$ we call a map $f:G^p\rightarrow N$ \emph{locally smooth} if there exists an open identity neighbourhood $U\subseteq G^p$ such that $f_{\mid U}$ is smooth and we say that $f$ is \emph{normalized} if
\[(\exists j)\,g_j= 1_G \quad \Rightarrow \quad f(g_1,\ldots,g_p)=1_N.
\]We denote by $C^p_s(G,N)$ the space of all normalized locally smooth maps $G^p\rightarrow N$, the so called (locally smooth) p-$\emph{cochains}$. For $p=2$ and non-connected Lie groups, we sometimes have to require additional smoothness (cf. Appendix \ref{section extension smooth crossed module}). We write $C^2_{ss}(G,N)$ for the set of all elements $\omega\in C^2_{s}(G,N)$ with the additional property that for each $g\in G$ the map
\[\omega_g:G\rightarrow N, \quad x\mapsto\omega(g,x)\omega(gxg^{-1},g)^{-1}
\]is smooth in an identity neighbourhood of $G$. At this point is is worth mentioning that we will repeatedly use the fact that $\sigma\in C^1_s(G,N)$ implies that the induced map $\delta_\sigma :G\times G\rightarrow N$ defined by $$\delta_\sigma(g,g'):=\sigma(g)\sigma(g')\sigma(gg')^{-1}$$ gives rise to an element of $C^2_{ss}(G,N)$. Moreover, we define $C^1_s(G,\Aut(N))$ as the set of all maps $S:G\rightarrow \Aut(N)$ with $S(1_G)=\id_N$ and for which there exists an open identity neighbourhood $U\subseteq G$ such that the induced map $S^{\wedge}$ is smooth on $U\times N$. For $S\in C^1_s(G,\Aut(N))$ and $\omega\in C^2_s(G,N)$ we put 
\begin{align}
(d_S\omega)(g,g',g''):=S(g)(\omega(g',g''))\omega(g,g'g'')\omega(gg',g'')^{-1}\omega(g,g{'})^{-1} \label{2-cocycle}
\end{align}
and call the elements of the set
\[Z^2_{ss}(G,N):=\{(S,\omega)\in C^1_s(G,\Aut(N))\times C^2_{ss}(G,N):\delta_S=C_N\circ\omega, \, d_S\omega=1_N\}
\]\emph{smooth factor systems} for the pair $(G,N)$. Here, the map $\delta_S:G\times G\rightarrow \Aut(N)$ is defined by $$\delta_S(g,g'):=S(g)S(g')S(gg')^{-1}.$$ Given such a smooth factor system $(S,\omega)$, we write $N\times_{(S,\omega)} G$ for the set $N\times G$ endowed with the group multiplication 
\begin{align}
(n,g)(n',g')=(nS(g)(n')\omega(g,g'),gg')\notag.
\end{align}
and we recall from \cite[Proposition 2.8]{Ne06} that if $N$ is connected, the $N\times_{(S,\omega)} G$ carries a unique structure of a Lie group for which the map
$N\times G\rightarrow N\times_{(S,\omega)} G$, $(n,g) \mapsto (n,g)$ is smooth on a set of the form $N\times U$, where $U$ is an open identity neighbourhood in $G$. Moreover, the map $q:N\times_{(S,\omega)} G\rightarrow G$, $(n,g)\mapsto g$ is a Lie group extension of $G$ by $N$.  In the case $N=A$ is an abelian Lie group the adjoint representation of $A$ is trivial and a smooth factor system $(S,\omega)$ for $(G,A)$ consists of a smooth module structure $S:G\rightarrow \Aut(A)$ and an element $$\omega\in Z^2_{ss}(G,A):=\{\omega\in C^2_{ss}(G,A): d_S\omega=1_A\}.$$

Last but not least, given a Lie group $G$, we write small fractured letters for its Lie algebra, i.e., $\mathfrak{g}$ for the Lie algebra of $G$. 


\section{On lifting of automorphisms}\label{LAAE}

Let $G$ be a Lie group, $(A,S)$ a smooth $G$-module and  $N$ a connected split normal Lie subgroup of $G$. Furthermore,  let
\[E: 1\longrightarrow A{\longrightarrow}\widehat{N}\stackrel{\widehat{q}}{\longrightarrow}N\longrightarrow1
\]be an abelian Lie group extension of $N$ by the abelian Lie group $A$ for which the associated $N$-module structure on $A$ is given by the restriction of the map $S$ to $N$ and 
\[\Aut(\widehat{N},A):=\{\varphi\in\Aut(\widehat{N}):\,\varphi(A)=A\}
\]the group of all Lie group automorphisms of $\widehat{N}$ preserving the split Lie subgroup $A$. The splitness condition implies that the map
\[\Phi:\Aut(\widehat{N},A)\rightarrow\Aut(A)\times\Aut(N), \quad \Phi(\varphi):=(\varphi_A,\varphi_N),
\]where $\varphi_A:=\varphi_{\mid A}$ denotes the restriction to $A$ and $\varphi_N$ is defined by $\varphi_N\circ \widehat{q}=\widehat{q}\circ\varphi$, is a well-defined group homomorphism. The purpose of this section it to provide conditions assuring that the smooth group homomorphism
\[\psi:G\rightarrow\Aut(A)\times\Aut(N),\quad \psi(g):=(S(g),c_g)
\]lifts to a smooth group homomorphism $\widehat{\psi}:G\rightarrow\Aut(\widehat{N},A)$. These considerations will later on be crucial for our goal of constructing a ``transgression" map
\[\tau:\text{H}^2_s(N,A)\rightarrow \text{H}^3_s(G/N,A^N).
\]Parts of the discussion in this section are inspired by Appendix A of \cite{Ne06}. To be more precise, for the sake of a concise presentation we adapted some of the results and proofs of Appendix A of \cite{Ne06} concerning automorphisms of general Lie group extensions to the special case of abelian Lie group extensions. New aspects are, for example, given by considering smooth structures, i.e., smooth group homomorphisms (cf. Lemma \ref{lifting condition group I}, Proposition \ref{lifting condition group II} and Proposition \ref{lifting condition group III}). We start with the following characterization of the kernel of the map $\Phi$ defined above:

\begin{lemma}\label{ker Phi}
The map 
\[\Psi: Z^1_s(N,A)\rightarrow\ker(\Phi), \quad \Psi(f):=(f\circ \widehat{q})\cdot\id_{\widehat{N}}
\]is a bijective group homomorphism.
\end{lemma}
\begin{proof}
We first show that $\Psi$ is well-defined. Therefore, let $f\in Z^1_s(N,A)$. Then it is not hard to see that $\Psi(f)$ is a diffeomorphism of $\widehat{N}$ satisfying $\Psi(f)_{\mid A}=\id_A$ and $\widehat{q}=\widehat{q}\circ\Psi(f)$ whose inverse is given by $\Psi(-f)$. That it is also a group homomorphism follows for $n:=\widehat{q}(\widehat{n})$ and $n':=\widehat{q}(\widehat{n}')$ from
\begin{align*}
\Psi(f)(\widehat{n}\widehat{n}')&=f(nn')\widehat{n}\widehat{n}'=f(n)n.f(n')\widehat{n}\widehat{n}'=f(n)\widehat{n}f(n')\widehat{n}^{-1}\widehat{n}\widehat{n}'\\
&=f(n)\widehat{n}f(n')\widehat{n}'=\Psi(f)(\widehat{n})\Psi(f)(\widehat{n}').
\end{align*}
Furthermore, it is easily checked that $\Psi$ is an injective group homomorphism. To verify the surjectivity of $\Psi$, we choose $\varphi\in\ker(\Phi)$. Then there exists a smooth function $f:N\rightarrow A$ with $\varphi=(f\circ \widehat{q})\cdot\id_{\widehat{N}}$ and the fact that $\varphi$ is a group homomorphism implies that $f\in Z^1_s(N,A)$.
\end{proof}

Our next goal is to characterize the image of the map $\Phi$. To this end we constantly use the fact that the group $\Aut(A)\times\Aut(N)$ acts naturally on the space  $C^1_s(N,\Aut(A))\times C^p_s(N,A)$ by
\[(\varphi_A,\varphi_N).(S,f):=(c_{\varphi_A}\circ S\circ\varphi_N^{-1},\varphi_A\circ f\circ(\varphi_N^{-1}\times\cdots\times\varphi_N^{-1})).
\]In particular, a short calculation shows that for  $f\in C^p_s(N,A)$ and $f':=(\varphi_A,\varphi_N).f$ we have 
\[d_Nf'=\varphi_A\circ(d_Nf)\circ(\varphi_N^{-1}\times\cdots\times \varphi_N^{-1})
\]which implies that the subset $Z^p_s(N,A)$ is invariant under the action of the group $\Aut(A)\times\Aut(N)$.

\begin{lemma}\label{natural action}
Let $(\varphi_A,\varphi_N)\in \Aut(A)\times\Aut(N)$ and consider the exact sequence 
\[E':=(\varphi_A,\varphi_N).E: 1\longrightarrow A\stackrel{\varphi_A^{-1}}{\longrightarrow}\widehat{N}\stackrel{\varphi_N\circ \widehat{q}}{\longrightarrow}N\longrightarrow1.
\]If $\widehat{N}=A\times_{(S,f)} N$ for some locally smooth cocycle $f\in Z^2_s(N,A)$, then the sequence $E'$ describes an extension equivalent to $A\times_{(\varphi_A,\varphi_N).(S,f)}N$ and the map
\[\mu_{(\varphi_A,\varphi_N)}:A\times_{(S,f)}N\rightarrow A\times_{(\varphi_A,\varphi_N).(S,f)}N, \quad (a,n)\mapsto(\varphi_A(a),\varphi_N(n))
\]is an isomorphism of Lie groups.
\end{lemma}
\begin{proof}
Let $\sigma:N\rightarrow A\times_f N$, $n\mapsto(1_A,n)$ be the canonical section. Then the map $\sigma':=\sigma\circ\varphi_N^{-1}\in C^1_s(N,\widehat{N})$ satisfies
\[(\varphi_N\circ \widehat{q})\circ\sigma'=\varphi_N\circ \widehat{q}\circ\sigma\circ\varphi_N^{-1}=\id_N.
\]As a consequence, we can interpret $\sigma'$ as a section for the extension $E'$ which implies that the sequence $E'$ can be represented by the extension $A\times_{(S',f')} N$ with $(S',f'):=(\varphi_A,\varphi_N).(S,f)$. That the map $\mu_{(\varphi_A,\varphi_N)}$ is an isomorphism of Lie groups follows from a short calculation involving the definition of $(\varphi_A,\varphi_N).(S,f)$.
\end{proof}

\begin{proposition}\label{lifting condition group}
Let $f\in Z^2_s(N,A)$ and $\widehat{N}=A\times_{(S,f)}N$ be the corresponding abelian extension. Furthermore, let $(\varphi_A,\varphi_N)\in\Aut(A)\times\Aut(N)$. Then the following assertions hold:
\begin{itemize}
\item[(a)]
The pair $(\varphi_A,\varphi_N)$ is contained in the image of $\Phi$, i.e., $(\varphi_A,\varphi_N)\in\im(\Phi)$ if and only if the extensions $\widehat{N}$ and $A\times_{(\varphi_A,\varphi_N).(S,f)} N$ are equivalent.
\item[(b)]
An automorphism $\varphi\in\Aut(\widehat{N},A)$ satisfies $\Phi(\varphi)=(\varphi_A,\varphi_N)$ if and only if it is of the form
\[\varphi(a,n)=(\varphi_A(a)+h(\varphi_N(n)),\varphi_N(n))
\]with $h\in C_s^1(N,A)$ satisfying $(\varphi_A,\varphi_N).(S,f):=(S,f+d_Nh)$.
\end{itemize}
\end{proposition}
\begin{proof}
(a) Each element $\varphi\in\Aut(\widehat{N},A)$ gives rise to an equivalence of extensions of the form
\[\xymatrix{ 1 \ar[r]& A\ar[r]^{\varphi_A^{-1}} \ar@{=}[d]& \widehat{N}\ar[r]^{\varphi_N\circ\widehat{q}} \ar[d]^{\varphi}& N\ar[r] \ar@{=}[d]& 1 \\
1 \ar[r]& A\ar[r] & \widehat{N}\ar[r]^{\widehat{q}} & N \ar[r] & 1,}
\]where $(\varphi_A,\varphi_N)=\Phi(\varphi)$. The equivalence of $\widehat{N}$ and \mbox{$A\times_{(\varphi_A,\varphi_N).(S,f)}N$} is therefore a consequence of Lemma \ref{natural action}. If, conversely, $\widehat{N}$ and $A\times_{(\varphi_A,\varphi_N).(S,f)}N$ are equivalent, then it follows from Lemma \ref{natural action} that there exists an equivalence of extensions of the form
\[\xymatrix{ 1 \ar[r]& A\ar[r]^{\varphi_A^{-1}} \ar@{=}[d]& \widehat{N}\ar[r]^{\varphi_N\circ \widehat{q}} \ar[d]^{\varphi}& N\ar[r] \ar@{=}[d]& 1 \\
1 \ar[r]& A\ar[r] & \widehat{N}\ar[r]^{\widehat{q}} & N \ar[r] & 1,}
\]which in turn implies that $\Phi(\varphi)=(\varphi_A,\varphi_N)$.

(b) Let $\phi\in\Aut(\widehat{N},A)$ with $\Phi(\varphi)=(\varphi_A,\varphi_N)$. Then we have
\[\varphi(a,n)=(\varphi_A(a)+h(\varphi_N(n)),\varphi_N(n))
\]for some $h\in C^1_s(N,A)$. Moreover, if $(S',f'):=(\varphi_A,\varphi_N).(S,f)$ and 
\[\mu_{(\varphi_A,\varphi_N)}:A\times_{(S,f)} N\rightarrow A\times_{(S',f')}N,\quad (a,n)\mapsto(\varphi_A(a),\varphi_N(n))
\]is the isomorphism of Lie groups from Lemma \ref{natural action}, then we obtain an isomorphism
\[\nu:=\varphi\circ\mu_{(\varphi_A,\varphi_N)}^{-1}:A\times_{(S',f')} N\rightarrow A\times_{(S,f)}N,\quad (a,n)\mapsto (a+h(n),n)
\]which is an equivalence of extensions. Consequently, $(S',f')=(S,f+d_Nh)$. If, conversely, $$(S',f'):=(\varphi_A,\varphi_N).(S,f)=(S,f+d_Nh),$$ then $\nu$ is an equivalence of extensions and $\varphi:=\nu\circ \mu_{(\varphi_A,\varphi_N)}$ an automorphism of Lie groups satisfying $\Phi(\varphi)=(\varphi_A,\varphi_N)$.
\end{proof}

Summarizing the previous results, we have shown the following statement:

\begin{theorem}
Let $(\Aut(A)\times\Aut(N))_{[E]}$ be the subgroup of $\Aut(A)\times\Aut(N)$ preserving the extension class $[E]\in\Ext(N,A)$. Then the following sequence of groups is exact:
\[1\longrightarrow Z^1_s(N,A)\stackrel{\Psi}{\longrightarrow}\Aut(\widehat{N},A)\stackrel{\Phi}{\longrightarrow}(\Aut(A)\times\Aut(N))_{[E]}\longrightarrow 1.
\]
\end{theorem}


In what follows we denote by $S_{\psi}$ the action of the group $G$ on $C^1_s(N,A)$ given by $g.f:=\psi(g).f$ and write $(C^p_s(G,C^1_s(N,A)),d_{S_{\psi}})$ for the corresponding chain complex consisting of maps $f:G^p\rightarrow C^1_s(N,A)$ for which the induced map 
\[G^p\times N\rightarrow A, \quad (g_1,\ldots,g_p,n)\mapsto f(g_1,\ldots,g_p)(n)
\]is smooth in an identity neighbourhood of $G^p\times N$.

\begin{definition}\label{smthcohoinv}(\cite[Definition D.5]{Ne04}). We call an element $f\in Z^2_s(N,A)$ \emph{smoothly cohomologically invariant} if there exists a map $\theta\in C^1_s(G,C^1_s(N,A))$ satisfying 
\begin{align*}
d_N(\theta(g))=g.f-f \quad \text{for all} \quad g\in G.
\end{align*}
In particular, we write $Z^2_s(N,A)^{[G]}$ for the set of smoothly cohomologically invariant cocycles in the group $Z^2_s(N,A)$. A few moments thought then shows that $B^2_s(N,A)\subseteq Z^2_s(N,A)^{[G]}$. In fact, if $f=d_Nh$ for some $h\in C^1_s(N,A)$, then we may choose $\theta(g):=g.h-h$. We can therefore define the space of \emph{smoothly invariant cohomology classes} by
\[\text{H}^2_s(N,A)^{[G]}:=Z^2_s(N,A)^{[G]}/B^2_s(N,A).
\]
\end{definition}

\begin{lemma}
\label{lifting condition group I}
Let $f\in Z^2_s(N,A)^{[G]}$ and $\widehat{N}:=A\times_{(S,f)} N$ be the corresponding abelian Lie group extension. Furthermore, let $\theta$ be as in Definition \ref{smthcohoinv}. Then the following assertions hold:
\begin{itemize}
\item[(a)]
For each $g\in G$ the map $\widehat{\psi}_{\theta}(g):\widehat{N}\rightarrow\widehat{N}$ defined by
\[\widehat{\psi}_{\theta}(g)(a,n):=\bigl(S(g)(a)+\theta(g)(c_g(n)),c_g(n)\bigr)
\]defines an element of $\Aut(\widehat{N},A)$ with $\Phi(\widehat{\psi}_{\theta}(g))=\psi(g)$. In particular, we have $\im(\psi)\subseteq \Aut(A)\times\Aut(N))_{[E]}$.
\item[(b)]
For each $g\in G$ the map $\widehat{\phi}_{\theta}(g):\widehat{N}\rightarrow\widehat{N}$ defined by
\[\widehat{\phi}_{\theta}(g)(a,n):=\bigl(S(g^{-1})(a)-g^{-1}.\theta(g)(c_{g^{-1}}(n)),c_{g^{-1}}(n)\bigr)
\]is the inverse of the automorphism $\widehat{\psi}_{\theta}(g)$. 
\item[(c)]
The map $\delta_{\widehat{\psi}_{\theta}}:G\times G\rightarrow \ker(\Phi)$, $\delta_{\widehat{\psi}_{\theta}}(g,g'):=\widehat{\psi}_{\theta}(g)\widehat{\psi}_{\theta}(g')\widehat{\psi}_{\theta}(gg')^{-1}$ satisfies
\[\delta_{\widehat{\psi}_{\theta}}(g,g')=(d_{S_{\psi}}\theta(g,g')\circ\widehat{q})\cdot\id_{\widehat{N}}.
\]In particular, the element $d_{S_{\psi}}\theta\in C^2_s(G,C^1_s(N,A))$ has values in $Z^1_s(N,A)$. 
\item[(d)]
The map $\widehat{\psi}_{\theta}:G\rightarrow\Aut(\widehat{N},A)$, $g\mapsto\widehat{\psi}_{\theta}(g)$ is a smooth group homomorphism if and only if the map $\theta$ is a locally smooth 1-cocycle, i.e., if $\theta\in Z^1_s(G,C^1_s(N,A))_{S_{\psi}}$.
\end{itemize}
\end{lemma}
\begin{proof}
(a) The first assertion is a direct consequence of Proposition \ref{lifting condition group} (b) applied for each $g\in G$ to the locally smooth 1-cochain $\theta(g)\in C^1_s(N,A)$.

(b) The second assertion immediately follows from a direct calculation. Indeed, we have for all $g\in G$ and $(a,n)\in\widehat{N}$
\begin{align*}
\bigl(\widehat{\psi}_{\theta}(g)\circ\widehat{\phi}_{\theta}(g)\bigr)(a,n)&=\widehat{\psi}_{\theta}(g)\bigl(S(g^{-1})(a)-g^{-1}.\theta(g)(c_{g^{-1}}(n)),c_{g^{-1}}(n)\bigr)\\
&=\bigl(a-\theta(g)(n)+\theta(g)(n),n\bigr)=(a,n)
\end{align*}
and 
\begin{align*}
\bigl(\widehat{\phi}_{\theta}(g)\circ\widehat{\psi}_{\theta}(g)\bigr)(a,n)&=\widehat{\phi}_{\theta}(g)\bigl(S(g)(a)+\theta(g)(c_g(n)),c_g(n)\bigr)\\
&=\bigl(a+S(g^{-1})(\theta(g)(c_g(n))-g^{-1}.\theta(g)(n)),n\bigr)=(a,n).
\end{align*}

(c) We first note that we have for all $g,g'\in G$ and $(a,n)\in\widehat{N}$
\begin{align*}
\widehat{\psi}_{\theta}(g)\widehat{\psi}_{\theta}(g')(a,n)&=\widehat{\psi}_{\theta}(g)\bigl(S(g')(a)+\theta(g')(c_{g'}(n)),c_{g'}(n)\bigr)\\
&=\bigl(S(gg')(a)+S(g)\theta(g')(c_{g'}(n))+\theta(g)(c_{gg'}(n)),c_{gg'}(n)\bigr)\\
&=\bigl(S(gg')(a)+(g.\theta(g')+\theta(g))(c_{gg'}(n)),c_{gg'}(n)\bigr).
\end{align*}
Furthermore, part (b) of the statement implies that
\begin{align*}
\widehat{\psi}_{\theta}(gg')^{-1}(a,n)=\bigl(S((gg')^{-1})(a)-(gg')^{-1}.\theta(gg')(c_{(gg')^{-1}}(n)),c_{(gg')^{-1}}(n)\bigr).
\end{align*}
A short calculation therefore shows that
\begin{align*}
\delta_{\widehat{\psi}_{\theta}}(g,g')(a,n)&=\bigl(a-\theta(gg')(n)+(g.\theta(g')+\theta(g)(n),n)\\
&=(d_{S_{\psi}}\theta(g,g')(n)+a,n\bigr),
\end{align*}
which in turn gives the desired equality $$\delta_{\widehat{\psi}_{\theta}}(g,g')=(d_{S_{\psi}}\theta(g,g')\circ\widehat{q})\cdot\id_{\widehat{N}}.$$
That the element $d_{S_{\psi}}\theta\in C^2_s(G,C^1_s(N,A))$ has values in $Z^1_s(N,A)$ is a consequence of Lemma \ref{ker Phi}.

(d) The formula in part (c) shows that the map $\widehat{\psi}_{\theta}$ is a group homomorphism if and only if $\theta\in Z^1_s(G,C^1_s(N,A))_{S_{\psi}}$. The smoothness of $\widehat{\psi}_{\theta}$ is in this case automatic and follows from \cite[Proposition A.13]{Ne06}.
\end{proof}

\begin{proposition}\label{lifting condition group II}
Suppose that we are in the situation of Lemma \ref{lifting condition group I}. Then the map $d_{S_{\psi}}\theta$ is the locally smooth 2-cocycle of the abelian extension $\psi^*(\Aut(\widehat{N},A))$ of $G$ by $Z^1_s(N,A)$ obtained by pulling back the abelian extension \[1\longrightarrow Z^1_s(N,A)\stackrel{\Psi}{\longrightarrow}\Aut(\widehat{N},A)\stackrel{\Phi}{\longrightarrow}(\Aut(A)\times\Aut(N))_{[E]}\longrightarrow1.
\]
\end{proposition}
\begin{proof}
We first note that the map 
\[\sigma_{\theta}:G\rightarrow\psi^*(\Aut(\widehat{N},A)), \quad g\mapsto(\widehat{\psi}_{\theta}(g),g)
\]defines a locally smooth normalized section of the pulled back extension, i.e., the induced maps
\[\sigma_{\theta}^{\wedge}:G\times\widehat{N}\rightarrow\widehat{N}\times G, \quad (g,(a,n))\mapsto(\widehat{\psi}_{\theta}(g)(a,n),g)
\]and
\[(\sigma_{\theta}^{-1})^{\wedge}:G\times\widehat{N}\rightarrow\widehat{N}\times G, \quad (g,(a,n))\mapsto(\widehat{\phi}_{\theta}(g)(a,n),g^{-1})
\]are smooth on $U\times \widehat{N}$ for some open identity neighbourhood $U\subseteq G$. Moreover, a few moments thought shows that 
\begin{align*}
\delta_{\sigma_{\theta}}(g,g'):=\sigma_{\theta}(g)\sigma_{\theta}(g')\sigma_{\theta}(gg')^{-1}=(\delta_{\widehat{\psi}_{\theta}}(g,g'),1_G).
\end{align*}
holds for all $g,g'\in G$. The statement therefore follows from Lemma \ref{lifting condition group I} (c).
\end{proof}

\begin{proposition}\label{lifting condition group III}
Let $f\in Z^2_s(N,A)$ and $\widehat{N}:=A\times_{(S,f)} N$ be the corresponding abelian Lie group extension. Then the smooth group homomorphism $\psi$ lifts to a smooth homomorphism $\widehat{\psi}:G\rightarrow\Aut(\widehat{N},A)$ if and only if $f\in Z^2_s(N,A)^{[G]}$ and the corresponding cohomology class
\[[d_{S_{\psi}}\theta]=[\psi^*(\Aut(\widehat{N},A))]\in\emph{H}^2_s(G,Z^1_s(N,A))_{S_{\psi}}
\]vanishes \emph{(}cf. Proposition \ref{lifting condition group II}\emph{)}.
\end{proposition}
\begin{proof}
Suppose first that $f\in Z^2_s(N,A)^{[G]}$ and that the corresponding cohomology class $[d_{S_{\psi}}\theta]$ vanishes. Then there exists an element $\theta'\in C^1_s(G,Z^1_s(N,A))$ with $d_{S_{\psi}}\theta=d_{S_{\psi}}\theta'$ which shows that the map $\vartheta:=\theta-\theta'$ defines an element in $C^1_s(G,C^1_s(N,A))$ satisfying $d_{S_{\psi}}\vartheta=0$, i.e., that $\vartheta\in Z^1_s(G,C^1_s(N,A))_{S_{\psi}}$. Moreover, a short calculations gives
\[d_N(\vartheta(g))=d_N(\theta(g))=g.f-f
\]for all $g\in G$. Therefore Lemma \ref{lifting condition group I} implies that $\widehat{\psi}:=\widehat{\psi}_{\vartheta}:G\rightarrow\Aut(\widehat{N},A)$ is a smooth group homomorphism lifting $\psi$. If, conversely, the smooth group homomorphism $\psi$ lifts to a smooth group homomorphism $\widehat{\psi}:G\rightarrow\Aut(\widehat{N},A)$, then it follows from Propositon \ref{lifting condition group} and \cite[Proposition A.13 (1)]{Ne06} that there exists a map $\theta\in C^1_s(G,C^1_s(N,A))$ as in Definition \ref{smthcohoinv} with $\widehat{\psi}=\widehat{\psi}_{\theta}$. In particular, we conclude that $f\in Z^2_s(N,A)^{[G]}$. Lemma \ref{lifting condition group I} (c) therefore implies that $\theta\in Z^1_s(G,C^1_s(N,A))_{S_{\psi}}$ which in turn means that $d_{S_{\psi}}\theta$ vanishes.
\end{proof}

\begin{remark}
It is not hard to see that each locally smooth 1-cocycle $f\in Z^1_s(N,A)$ actually defines a globally smooth map from $N$ to $A$ (cf. \cite[Lemma III.1]{Ne04}).
We further recall, that if the action of $N$ on $A$ is trivial, i.e., if we are in the context of central extensions, then $Z^1_s(N,A)=\Hom_{\text{gr}}(N,A)$.
\end{remark}



\begin{example}\label{action on loop groups}
Let $H$ be a compact, 1-connected and semisimple Lie group and 
\[G=PH:=\{\xi\in C^{\infty}(\mathbb{R},H):\,\xi(0)=1_H \,\,\,\text{and}\,\,\,\xi^{-1}d\xi \,\,\, 2\pi\text{-periodic}\}
\]the path group of $H$. Furthermore, let $N$ be the split normal Lie subgroup of $G$ realized as the based loop group $$\Omega H:=\{\xi\in C^{\infty}(S^1,H):\,\xi(1)=1_H\}$$ of $H$. It is a well-known fact, that the space $PH$ is a contractible $\Omega H$-principal bundle over $H$, hence $\pi_n(\Omega H)=\pi_{n+1}(H)$ holds for each $n\in\mathbb{N}_0$. In particular, $\Omega H$ is 1-connected. Moreover, \cite[Proposition 3.4.1]{PrSe86} implies that its Lie algebra $L\mathfrak{h}$ is topologically perfect, that is, the commutator algebra is dense in $L\mathfrak{h}$. If we identify $S^1$ with $\mathbb{R}/2\pi\mathbb{Z}$, i.e., functions on $S^1$ with $2\pi$-periodic functions on $\mathbb{R}$, then
\[\omega(\zeta,\eta):=\int_{[0,2\pi]} <\zeta(t),\eta'(t)> dt,
\]where $\zeta,\eta\in L\mathfrak{h}$ and $<\cdot, \cdot>$ is an invariant bilinear form on $\mathfrak{k}$, defines a continuous $\mathbb{R}$-valued 2-cocycle on $L\mathfrak{h}$ which is unique in the sense that any complex valued cocycle up to coboundaries is a multiple of $\omega.$ In particular, a short calculation shows that the map
\[\lambda:G\rightarrow\Lin_c(L\mathfrak{h},\mathbb{R}),\quad \lambda(g)(\zeta):=\int_{[0,2\pi]}
<g^{-1}(t)g'(t),\zeta(t)>dt
\]
is smooth and satisfies $\omega(g.\zeta,g.\eta)-\omega(\zeta,\eta)=\lambda(g)([\zeta,\eta])$ for all $g\in G$ and $\zeta,\eta\in L\mathfrak{h}$. Integrating this Lie algebra 2-cocycle therefore gives rise to a smoothly cohomologically invariant element in $Z^2_s(\Omega H,S^1)$. Since the group $Z^1_s(\Omega H,S^1)=\Hom_{\text{gr}}(\Omega H,S^1)$ is also trivial, we finally conclude from Proposition \ref{lifting condition group III} that the group homomorphism $\psi=C_{\Omega H}:G\rightarrow\Aut(\Omega H)$ lifts to a homomorphism $\widehat{\psi}:G\rightarrow\Aut(\widehat{N},S^1)$ inducing a smooth action of $G$ on any central extension $\widehat{N}$ of $\Omega H$ by the circle $S^1$ (cf. \cite[Remark II.21]{Ne01}).

\end{example}

\section{The transgression map}\label{transgression section}

Let $G$ be a Lie group, $N$ a connected split normal Lie subgroup of $G$ with quotient map $q:G\rightarrow G/N$ and $(A,S)$ a smooth $G$-module. At this point we again recall that the splitness condition means that the quotient $G/N$ has a natural Lie group structure such that $q:G\rightarrow G/N$ defines on $G$ the structure of a principal $N$-bundle. Since $N$ is a normal subgroup of $G$, the subgroup 
\[A^N:=\{a\in A:\,(\forall n\in N)\,S(n)(a)=a\}
\]is a $G$-submodule of $A$. In what follows we additionally assume that $A^N$ is a split Lie subgroup of $A$. In this case it is not hard to check that the map
\begin{align*}
\Phi_A:\Aut_{A^N}(A)\rightarrow\Aut(A/A^N),\quad \Phi_A(\varphi)([a]):=[\varphi(a)]
\end{align*}
is a well-defined group homomorphism. Moreover, \cite[Lemma C.2]{Ne04} implies that $A^N$ inherits a natural structure of a smooth $G/N$-module (cf. Lemma \ref{basic staff on crossed module} below). The aim of this section is to construct a ``transgression" map
\[\tau:\text{H}^2_s(N,A)\rightarrow \text{H}^3_s(G/N,A^N).
\]Our strategy involves the theory of smooth crossed modules. We therefore start with a definition and some basic facts. A detailed background on smooth crossed modules, e.g., the construction of associated characteristic classes, can be found in \cite[Section III]{Ne06}.

\begin{definition}\label{crossed modules}
A pair $(\alpha,\widehat{S})$, consisting of a morphism $\alpha:\widehat{N}\rightarrow G$ of Lie groups together with a homomorphism $\widehat{S}:G\rightarrow\Aut(\widehat{N})$ defining a smooth action of $G$ on $\widehat{N}$, is called a $\emph{smooth crossed module}$ if the following conditions are satisfied:
\begin{itemize}
\item[(CM1)]
$\alpha\circ\widehat{S}(g)=c_g\circ\alpha$ for $g\in G$.
\item[(CM2)]
$\widehat{S}\circ\alpha=C_{\widehat{N}}:\widehat{N}\rightarrow\Aut(\widehat{N})$ is the conjugation action.
\item[(CM3)]
$\ker(\alpha)$ is a split Lie subgroup of $\widehat{N}$ and $\im(\alpha)$ is a split Lie subgroup of $G$ for which $\alpha$ induces an isomorphism $\widehat{N}/\ker(\alpha)\rightarrow\im(\alpha)$.
\end{itemize}
\end{definition}

\pagebreak[3]

\begin{lemma}\label{basic staff on crossed module}
Let $(\alpha,\widehat{S})$ be a smooth crossed module. Then the following assertions hold:
\begin{itemize}
\item[(a)]
$N:=\text{im}(\alpha)$ is a normal subgroup of $G$.
\item[(b)]
$Z:=\ker(\alpha)\subseteq Z(\widehat{N})$.
\item[(c)]
$Z$ is $G$-invariant and $G$ acts smoothly on $Z$. 
\item[(d)]
There is a smooth action $T:G/N\rightarrow\Aut(Z)$ given by $T([g]).z:=\widehat{S}(g)(z)$.
\end{itemize}
\end{lemma}
\begin{proof}
The first assertion follows from (CM1), while part (b) follows from (CM2). That $Z$ is $G$-invariant is a direct consequence of (CM1). For the proof that $G$ acts smoothly on $Z$ we refer to \cite[Lemma 1.7(a) and Lemma II.3]{Ne06}. To verify the last statement we observe that inner automorphisms of $\widehat{N}$ act trivially on $Z(\widehat{N})$ and thus on $\widehat{N}$. Therefore the smooth action of $G$ on $Z$ from part (c) factors through a smooth action $T:G/N\rightarrow\Aut(Z)$ given by $T([g]).z:=\widehat{S}(g)(z)$.
\end{proof}

We continue with an element $f\in Z^2_s(N,A)^{[G]}$, the corresponding abelian Lie group extension $\widehat{N}:=A\times_{(S,f)} N$ and the induced extension $\Gamma:=\psi^*(\Aut(\widehat{N},A))$ of $G$ by $Z^1_s(N,A)$ from Proposition \ref{lifting condition group II}. Then the group homomorphisms
\begin{align}
d_N:A\rightarrow Z^1_s(N,A), \quad d_N(a)(n):=n.a-a,\label{1-coboundary map}
\\
\alpha:\widehat{N}\rightarrow \Gamma,\quad \alpha(a,n):=(c_{(a,n)},n)\label{mapmap}
\end{align}
give rise to the following commuting diagram:
\[\xymatrix{ 1 \ar[r]& A\ar[r] \ar[d]^{\partial} & \widehat{N} \ar[r] \ar[d]^{\alpha} & N\ar[r] \ar[d]^{i}& 1 \\
1 \ar[r]& Z^1_s(N,A)\ar[r] & \Gamma \ar[r] & G \ar[r] & 1.}
\]Taking cokernels then results in the following short exact sequence of groups:
\begin{align}
1\longrightarrow \emph{H}^1_s(N,A){\longrightarrow}\Gamma/\alpha(\widehat{N}){\longrightarrow}G/N\longrightarrow1.\label{H^1}
\end{align}

\begin{lemma}\label{actions+kernel}
With the previous notations we have that
\[\widehat{S}:\Gamma\rightarrow\Aut(\widehat{N}),\quad \widehat{S}\bigr(\varphi,g\bigl)(a,n):=\varphi(a,n)
\]defines an action of $\Gamma$ on $\widehat{N}$. Moreover, $\ker(\alpha)=A^N$.
\end{lemma}
\begin{proof}
That the map $\widehat{S}$ defines an action of $\Gamma$ on $\widehat{N}$ is obvious. Therefore let $(a,n)\in\widehat{N}$ with $\alpha(a,n)=(c_{(a,n)},n)=(\id_{\widehat{N}},1_N)$. Then $n=1_N$ and for arbitrary $(a',n')\in\widehat{N}$ we compute
$c_{(a,1_N)}(a',n')=(a+a'-n'.a,n')$, which shows that $c_{(a,1_N)}(a',n')=(a',n')$ holds for all $(a',n')\in\widehat{N}$ if and only if $a\in A^N$. In particular, we conclude that $\ker(\alpha)=A^N$.
\end{proof}

Since every discrete group can be viewed as a Lie group, our discussion also includes the algebraic setting of crossed modules. In particular, if we consider all groups which appear in Lemma \ref{actions+kernel} as discrete groups, then it is a consequence of the construction that the pair $(\alpha,\widehat{S})$ defines a crossed module. Moreover, in this situation, it is possible to associate to the crossed module $(\alpha,\widehat{S})$ a characteristic class $\chi(\alpha,\widehat{S})\in \text{H}^3(\Gamma/\alpha(\widehat{N}),A^N)_T$ (cf. \cite[Lemma III.6]{Ne06}). In what follows we try to transfer the previous discussion to the smooth context. 


\begin{lemma}\label{splitstructureA}
Let $f\in Z^2_s(N,A)$ and $\widehat{N}:=A\times_{(S,f)} N$ be the corresponding abelian Lie group extension. Then $A^N$ is also a split Lie subgroup of $\widehat{N}$. In particular, the quotient $\widehat{N}/A^N$ has a natural Lie group structure such that the projection map $p:\widehat{N}\rightarrow\widehat{N}/A^N$ defines on $\widehat{N}$ the structure of a principal $A^N$-bundle.
\end{lemma}
\begin{proof}
In order to verify the assertion it is enough to show that $\widehat{N}/A^N$ has a natural Lie group structure which makes the following short exact sequence
\begin{align}
1\longrightarrow A/A^N{\longrightarrow}\widehat{N}/A^N{\longrightarrow}N\longrightarrow 1 \label{liegroupextI}
\end{align}
a Lie group extension of $N$ by $A/A^N$. We therefore consider the locally smooth section $$\sigma:N\rightarrow\widehat{N}, \quad n\mapsto(1_A,n)$$ of the extension $\widehat{q}:\widehat{N}\rightarrow A$, the induced locally smooth 2-cocycle $\delta_{\sigma}:N\times N\rightarrow A$ defined by $\delta_{\sigma}(n,n'):=\sigma(n)\sigma(n')\sigma(nn')^{-1}$ and the quotient map $p_A:A\rightarrow A/A^N$. Then it is not hard to check that $S=C_{A}\circ\sigma$ and that $f=\delta_{\sigma}$. Moreover, a few moments thought show that the pair $(\Phi_A\circ S,p_A\circ f)$, where
\begin{align*}
\Phi_A\circ S:N\rightarrow\Aut(A/A^N) \quad \text{and} \quad p_A\circ f:N\times N\rightarrow A/A^N,
\end{align*}
defines a smooth factor system for the pair $(N,A/A^N)$. In particular, we can use this smooth factor system to define a natural Lie group structure on $\widehat{N}/A^N$ which makes (\ref{liegroupextI}) into
a Lie group extension of $N$ by $A/A^N$.
\end{proof}

In the case the group $\mathrm{H}^1_s(N,A)$ vanishes, we can identify the space $Z^1_s(N,A)$ with the group $A/A^N$ (via the map in equation (\ref{1-coboundary map})). In particular, $\Gamma$ is as an abstract group an abelian extension of the group $G$ by the group $A/A^N$, i.e., the following sequence of groups is exact
\begin{align}
1\longrightarrow A/A^N{\longrightarrow}\Gamma\stackrel{\pr_2}{\longrightarrow}G\longrightarrow 1.\label{liegroupextII}
\end{align}

\begin{proposition}\label{liegroupstructureGamma}
Let $f\in Z^2_s(N,A)^{[G]}$ and $\widehat{N}:=A\times_{(S,f)} N$ be the corresponding abelian Lie group extension.
Furthermore, suppose that $\emph{H}^1_s(N,A)$ vanishes. Then $\Gamma$ carries a natural Lie group structure which makes the short exact sequence \emph{(}$\ref{liegroupextII}$\emph{)} an abelian Lie group extension of $G$ by $A/A^N$. Moreover, the action 
\[\widehat{S}^{\wedge}:\Gamma\times\widehat{N}\rightarrow\widehat{N},\quad \bigr((\varphi,g),(a,n)\bigl):=\varphi(a,n)
\]is smooth.
\end{proposition}
\begin{proof}
Since 
Since $f\in Z^2_s(N,A)^{[G]}$, we can find $\theta\in C^1_s(G,C^1_s(N,A))$ satisfying $d_N(\theta(g))=g.f-f$ for all $g\in G$ and a similar argument as in Lemma \ref{lifting condition group I} (c) shows that the element $d_{S_{\psi}}\theta\in C^2_s(G,C^1_s(N,A))$ has values in $A/A^N$. We claim that the pair $(\overline{S},\overline{\omega})$, where
\begin{align*}
\overline{S}:G\rightarrow\Aut(A/A^N), \quad \overline{S}(g)([a]):=[S(g)(a)]\\
\overline{\omega}:G\times G\rightarrow A/A^N, \quad \overline{\omega}(g,g'):=d_{S_{\psi}}\theta(g,g'),
\end{align*} 
defines a smooth factor system for the sequence (\ref{liegroupextII}). In fact, it follows along the sames lines as in the proof of Proposition \ref{lifting condition group II} that 
\[\sigma_{\theta}:G\rightarrow\Gamma, \quad g\mapsto(\widehat{\psi}_{\theta}(g),g)
\]defines a locally smooth normalized section of the projection map $\pr_2:\Gamma\rightarrow G$. In particular, a few moments thought show that the map $\overline{\omega}$ is the locally smooth 2-cocycle defining the abelian extension (\ref{liegroupextII}), that is, $\overline{\omega}\in Z^2_{ss}(G,A/A^N)$. Moreover, the corresponding $G$-module structure is given by the map $$C_{\Gamma}\circ\sigma_{\theta}:G\rightarrow\Aut(A/A^N)$$ which is for $g\in G$ and $a\in A$ given by
\begin{align*}
\bigl(\widehat{\psi}_{\theta}(g),g\bigr)\bigl((d_N(a)\circ\widehat{q})\cdot\id_{\widehat{N}},1_G\bigr)\bigl(\widehat{\psi}_{\theta}(g)^{-1},g^{-1}\bigr)=\bigl((d_N(S(g)(a))\circ\widehat{q})\cdot\id_{\widehat{N}},1_G\bigr),
\end{align*}
i.e., we conclude that $C_{\Gamma}\circ\sigma_{\theta}=\overline{S}$. The smoothness of the group homomorphism $\overline{S}$ is a consequence of the smoothness of the map $S$ and the fact that the natural quotient map $p_A:A\rightarrow A/A^N$ defines on $A$ the structure of a principal $A^N$-bundle. Hence, we have verified our claim and therefore we can use the smooth factor system $(\overline{S},\overline{\omega})$ to define a natural Lie group structure on $\Gamma$ which turns the sequence (\ref{liegroupextII}) into
a Lie group extension of $G$ by $A/A^N$. That the action 
\[\widehat{S}^{\wedge}:\Gamma\times\widehat{N}\rightarrow\widehat{N},\quad \bigr((\varphi,g),(a,n)\bigl):=\varphi(a,n)
\]is smooth is a consequence of \cite[Proposition A.13]{Ne06} because the map 
\[\Gamma\times N\rightarrow A, \quad ((\varphi,g),n)\mapsto\theta(g)(n)
\]is smooth in an identity neighbourhood.
\end{proof}

\begin{theorem}\label{Gammasmoothcrossedmodule}
Suppose we are in the situation of Proposition \ref{liegroupstructureGamma}. Then the pair $(\alpha,\widehat{S})$ defines a smooth crossed module. In particular, we obtain a ``transgression" map 
\[\tau:\emph{H}^2_s(N,A)^{[G]}_S\rightarrow \emph{H}^3_s(G/N,A^N)_T,\quad \tau([f]):=\chi(\alpha,\widehat{S}),
\]where $\chi(\alpha,\widehat{S})$ denotes the characteristic class of $(\alpha,\widehat{S})$ \emph{(}cf. \cite[Lemma III.6]{Ne06}\emph{)}.
\end{theorem}
\begin{proof}
By the previous results it remains to show that the group homomorphism $$\alpha:\widehat{N}\rightarrow\Gamma, \quad \alpha(a,n):=(c_{(a,n)},n)$$ is smooth and that $\im(\alpha)$ is a split Lie subgroup of $\Gamma$. Indeed, we first note that the following diagram of Lie groups commutes:
\[\xymatrix{ 1 \ar[r]& A/A^N\ar@{=}[d]\ar[r] & \widehat{N}/A^N \ar[r] \ar[d]^{j} & N\ar[r] \ar[d]^{i}& 1 \\
1 \ar[r]& A/A^N\ar[r] & \Gamma \ar[r]^{\pr_2} & G \ar[r] & 1.}
\]The smoothness of the map $j:\widehat{N}/A^N{\rightarrow}\Gamma$ therefore immediately follows from the smoothness of all the other appearing maps in the diagram. As a consequence, we conclude that the map $\alpha$ is smooth as a composition of the smooth projection map $p:\widehat{N}\rightarrow\widehat{N}/A^N$ (cf. Lemma \ref{splitstructureA}) and the smooth map $j:\widehat{N}/A^N{\rightarrow}\Gamma$. In order to see that $\im(\alpha)$ is a split Lie subgroup of $\Gamma$ it is enough to note that
\begin{align*}
1\longrightarrow\widehat{N}/A^N\stackrel{j}{\longrightarrow}\Gamma{\longrightarrow}G/N\longrightarrow 1
\end{align*}
is an extension of Lie groups.
\end{proof}

\begin{theorem}\label{extension}
Suppose that the quotient $G/N$ is connected and that the group $\emph{H}^1_s(N,A)$ vanishes. Furthermore, let $f\in Z^2_s(N,A)^{[G]}$. Then the following statements are equivalent:
\begin{itemize}
\item[(a)]
The element $\tau([f])\in \emph{H}^3_s(G/N,A^N)_T$ vanishes.
\item[(b)]
There exists an abelian Lie group extension $\widehat{G}$ of $G$ by $A$
extending 
\begin{align*}
1\longrightarrow A{\longrightarrow}\widehat{N}\longrightarrow N\longrightarrow 1.
\end{align*}
\end{itemize}
If these conditions are satisfied, then the set of different extensions is parametrized by the group $\emph{H}^2_{s}(G/N,A^N)_T$. 
\end{theorem}
\begin{proof}
Suppose first that the element $\tau([f])\in \text{H}^3_s(G/N,A^N)_T$ vanishes. Then it follows from Theorem \ref{Gammasmoothcrossedmodule} and \cite[Theorem III.8 (1) $\Rightarrow$ (2)]{Ne06} that there exists an abelian Lie group extension $\widehat{G}$ of $\Gamma$ by $A^N$
extending 
\begin{align*}
1\longrightarrow A^N{\longrightarrow}\widehat{N}\longrightarrow \widehat{N}/A^N\longrightarrow 1.
\end{align*}
Moreover, realizing $\widehat{G}$ and $G$ in product coordinates, we conclude from the first part of the proof of \cite[Theorem III.8]{Ne06} that the map
\begin{align*}
q_G:\widehat{G}=\widehat{N}\times(G/N)=& (A\times_{(S,f)}N)\times(G/N)\rightarrow G=N\times(G/N),
\\
&q_G((a,n),q(g)):=(n,q(g))
\end{align*}
defines a Lie group extension of $G$ by $A$ containing $\widehat{N}$. If conversely, there exists an abelian Lie group extension $\widehat{G}$ of $\Gamma$ by $A^N$
extending 
\begin{align*}
1\longrightarrow A{\longrightarrow}\widehat{N}\longrightarrow N\longrightarrow 1, 
\end{align*}
then a few moments thought shows that 
\begin{align*}
1\longrightarrow A^N{\longrightarrow}\widehat{G}\stackrel{\pr}\longrightarrow \Gamma\longrightarrow 1
\end{align*}
defines a central Lie group extension of $\Gamma$ by $A^N$. In particular, the canonical map $\widehat{N}\rightarrow{\pr}^{-1}(\widehat{N}/A^N)$ is a $\Gamma$-equivariant equivalence of $A^N$ extensions of $\widehat{N}/A^N$. It is therefore a consequence of Theorem \ref{Gammasmoothcrossedmodule} and \cite[Theorem III.8 (2) $\Rightarrow$ (1)]{Ne06} that the element $\tau([f])=\chi(\alpha,\widehat{S})\in \text{H}^3_s(G/N,A^N)_T$ vanishes.

Finally, if these conditions are satisfied, then it follows from Corollary \ref{H^2(G,Z)} that the set of different extensions is is parametrized by the group $\emph{H}^2_{s}(G/N,A^N)_T$.
\end{proof}



In the remaining part of this section we explain how the previous results can be used to interpret classes in locally smooth degree-3 group cohomology $\text{H}^3_s(H,S^1)$ as crossed modules associated to the loop group $\Omega H$ of a compact, 1-connected and semisimple Lie group $H$. Indeed, 
we first note that Example \ref{action on loop groups} implies that each element $f\in Z^2_s(\Omega H,S^1)$ gives rise to a smooth crossed module $(\alpha,\widehat{S})$, i.e., we have
\[Z^2_s(\Omega H,S^1)^{[PH]}=Z^2_s(\Omega H,S^1).
\]In particular, as a consequence of Theorem \ref{Gammasmoothcrossedmodule} we obtain a well-defined map
\[\tau_H:\text{H}^2_s(\Omega H,S^1)\rightarrow \text{H}^3_s(H,S^1),\quad [f]\mapsto\chi(\alpha,\widehat{S}).
\]That this map is a homomorphism follows from some basic calculations involving the Baer product of two extensions (cf. \cite[Remark II.15]{Ne06}).

\begin{theorem}\label{main theorem on crossed modules vs. loop groups}
Let $H$ be a semisimple, compact and 1-connected Lie group and consider $S^1$ as a trivial $H$-module. Then the map
\[\tau_H:\emph{H}^2_s(\Omega H,S^1)\rightarrow \emph{H}^3_s(H,S^1),\quad [f]\mapsto\chi(\alpha,\widehat{S}).
\]is an isomorphism. In particular, each element in $\emph{H}^3_s(H,S^1)$ can be realized as a characteristic class of a smooth crossed module associated to the loop group $\Omega H$.
\end{theorem}
\begin{proof}
We first note that we have injective derivation maps $D_2: \text{H}^2_s(\Omega H,S^1)\rightarrow\text{H}^2_s(\Omega\mathfrak{h},\mathbb{R})$ and $D_3:\text{H}^3_s(H,S^1)\rightarrow\text{H}^3_s(\mathfrak{h},\mathbb{R})$. In fact, the injectivity of the map $D_2$ is a consequence of \cite[Remark 7.5]{Ne01} while the injectivity of $D_3$ follows from \cite[Remark IV.16 and Remark V.14]{WW}. Next, we claim that we have a commuting square of maps
\begin{align}
\xymatrix{\text{H}^2_s(\Omega H,S^1)\ar[d]^{D_2}\ar[r]^{\tau_H} & \text{H}^3_s(H,S^1) \ar[d]^{D_3}  \\
 \text{H}^2_c(\Omega\mathfrak{h},\mathbb{R})\ar[r]^{t} & \text{H}^3_c(\mathfrak{h},\mathbb{R}),}
\end{align} 
where $\tau$ denotes the inverse of the Lie algebra transgression map from $\text{H}^3_c(\mathfrak{h},\mathbb{R})$ (interpreted as closed left invariant forms on $H$) to $\text{H}^2_c(\Omega\mathfrak{h},\mathbb{R})$. Indeed, the composite map $D_3\circ\tau_H$ is as follows: The locally smooth 3-cocycle representing $\tau_H([f])$ is precisely the obstruction to extending the central
extension of the loop group $\Omega H$ to the total space $PH$ of the fibration with base $H$ and fiber $\Omega H$. That is, the 3-cocycle is measuring the nonassociativity of the product in $\Omega H$ when extended to $PH$. On the Lie algebra level this is the breaking of the Jacobi identity when the 2-cocycle of $\Omega \mathfrak{h}$ is extended to $P\mathfrak{h}$. Moreover, the composite map $t\circ D_2$ is as follows: 
The $\mathbb{R}$-valued 2-cocycle on $\Omega\mathfrak{h}$ representing $D_2([f])$ is given by 
\begin{align*}
\omega_c(\zeta,\eta):=c\cdot\int_{[0,2\pi]}<(\xi(t),\eta'(t)>dt
\end{align*}
for some $c\in\mathbb{R}$ and  $\zeta,\eta\in \Omega\mathfrak{h}.$ The continuous extension of this to the path algebra $P\mathfrak{h}$ is given by the same formula, but since the paths are nonperiodic, there is a boundary term when trying to check the Jacobi identity. This boundary term is easily seen to be equal to
\begin{align*}
c\cdot <\zeta(2\pi),[\eta(2\pi),\xi(2\pi)]>
\end{align*}
for some $c\in\mathbb{R}$ and  $\zeta,\eta,\xi\in P\mathfrak{h}$, what proves the claim, i.e., $D_3\circ\tau_H=t\circ D_2$. We finally recall that at the infinitesimal level it is known that the central extensions which lead to extensions of the loop group are labelled by the integers on each
simple factor of $H.$  By the Lie algebras transgression map $t$ the integral extensions lead to 3-cocycles on $\mathfrak{h}$ which again are integral in the following sense: a Lie algebra cocycle is the same thing as a left invariant closed 3-form on $H$ and the integrality then means that the pairing with 3-cycles are integral. We therefore conclude that $\tau_H$ is an isomorphism since it can be written as a composition of isomorphisms:
\begin{align*}
\tau_H=\bigl(D_3^{\mid\text{H}^3_c(\mathfrak{h},\mathbb{Z})}\bigr)^{-1}\circ t\circ D_2
\end{align*}
\end{proof}

\section{More on lifting of automorphisms}

In this section we shall give a different construction of lifting the automorphisms of $N$ 
to automorphisms of an abelian extension $\hat N$ in the case when $A$ is the $G$-module of
smooth maps from $G$ to the unit circle $S^1.$  This case is important in the gauge theory applications.
As described in the introduction, the fibration $N \to \mathcal A \to \mathcal A/N$,
where $N= Map_0(S^n, K)$ (with $n=1,2,\dots$) is the group of based gauge transformations in a trivial
vector bundle over $S^n,$ and $\mathcal A$ is the space of smooth gauge connnections, is homotopy
equivalent to the fibration $N\to G \mapsto H$ for $G= PH,$ $H=Map_0(S^{n-1}, K)$ when $n\geq 2$
and $H=K$ for $n=1.$ 
In quantum field theory the group $N$ acts on the fermionic Fock space through an abelian extension
with fiber the group $Map(\mathcal A, S^1),$ here replaced by $Map_0(G, S^1).$ 

However, there is a problem with the above definitions. Since $G$ is in general a locally
convex Lie group it is not clear what is the smooth structure on the module $A$ to make it
into an abelian Lie group.  In practice, this is not a serious problem since in the
gauge theory applications the extension of $N$ by $A$ can be restricted to a subgroup
$A'$ of $A$ consisting only of a 'small' set of $S^1$-valued functions on $G.$ Typically
we can take $A$ as a certain set of differential polynomials in a gauge connection for
which a locally convex structure can be given in terms of the coefficients in the
differential polynomials which in turn are certain matrix valued functions on a compact
manifold \cite{Mi3}.  Another way to get around the smoothness problem in $A$ comes from
the realization that an extension of $N$ by $A$ above is the same as an $S^1$-extension
of the transformation groupoid $G \times N \to G.$ In this way a Lie group structure on
the $A$-extension of $N$ is replaced by the smoothness of the local 
2-cocycle $G\times N \times N \to S^1$ for the action groupoid $G\times N \to G.$

So again $G$ is an extension of a group $H$ by a normal subgroup $N$ of $G$. Fix an element
$\omega_3 \in \mathrm{H}^3(H, \Bbb Z).$ The (singular) cocycle $\omega_3$ will be used to
construct $\hat N.$ Assume that $\pi^*\omega_3 = d \theta_2$, where $\pi: G \to
H$ denotes the canonical projection. Then $\theta_2$ is closed along all $N$-orbits in $G.$ 

Let us first assume that $N$ is connected. We define an extension $\hat N$ of $N$ by the
abelian group $A$ of smooth functions on $G$ with values in the unit circle $S^1.$ 
Equivalently, we define an $S^1$-extension of the smooth transformation groupoid $G\times N \to G.$ 
Elements of $\hat N$ are equivalence classes of pairs $(g(\cdot),\lambda)$ with $g(t)$ a smooth
path in $G$ joining $g=g(0)$ to a point $g(1)= gn$ with $n\in N$ and $\lambda\in S^1.$
The equivalence relation is defined as $(g(\cdot),\lambda) \sim (g'(\cdot),\lambda')$ 
when $g(t) = g'(t)$ at $t=0,1$ and $\lambda' = \lambda \gamma(g*g_-'),$
where $g*g_-'$ is the composed loop from $g(0)$ to $g(1)$ and then travelling backwards from
$g(1) = g'(1)$ along $g'$ to $g'(0) = g(0).$ Here $\gamma(g*g_-')$ is the parallel transport along
the loop $g*g_-'$ in a circle bundle with the curvature $2\pi i \theta_2$ on the orbit $g(0)N;$ 
this is uniquely defined if $N$ is simply connected, and given by pairing $\theta_2$ with any
surface with boundary equal to $g*g_-'.$ In the case when $N$ is not simply connected we must
supplement the information in $\theta_2$ with a homomorphism $\pi_1(N) \to S^1;$ that is, 
we fix a differential character $\gamma$ \cite{CS} from the vertical 1-cycles in $G$ to $S^1$ with curvature 
$2 \pi i\theta_2.$ The groupoid multiplication is defined as follows. Take a pair of equivalence classes 
 represented respectively by the pairs $(g_1(\cdot), \lambda_1)$ and $(g_2(\cdot), \lambda_2).$
 The product is defined as the class represented by the pair $(g_1* g_2, \lambda_1\lambda_2).$ 
 It is easy to see (draw a diagram!) that this composition is well-defined and 
 associative on the level of equivalence classes. 
 
 Consider next the case when $N$ is disconnected. Denote $X= N/N_0$, where $N_0\subset N$ is
 the path connected component of the identity. We consider the quotient $X$ as a discrete group.
 The conjugation $n_0 \to n n_0 n^{-1}$ defines an automorphism of $N_0$ and also an automorphism
 of the transformation groupoid $G \times N_0 \to G$ by $(g,n_0) \mapsto (gn^{-1}, nn_0 n^{-1}).$
 Now we make an additional assumption: $L_n^* \theta_2 = \theta_2 + d\phi$, where $L_n: G\to G,
 L_n(g) = gn^{-1}.$ 
 
 \begin{proposition} Fix an element $\omega_3\in \mathrm{H}^3(H, \Bbb Z).$ With the notation above,
 in the category of smooth (locally convex) action groupoids and their $S^1$ extensions,
 the automorphim $Aut_n$ for $n\in N$ on the groupoid $G\times N_0 \to G$
 extends to an automorphism of the groupoid $S^1$-extension in the case $L_n^* \theta_2 = \theta_2 +d\phi.$
 \end{proposition}
 \begin{proof} The condition on $\theta_2$ above corresponds to the condition in Definition  
  \ref{smthcohoinv}.
 Let $(g(\cdot), \lambda)$ be a representative for the groupoid extension, with
 $g(1)= g(0) n_0$ and define
 $$ F_n(g(\cdot), \lambda) = (g(\cdot) n^{-1}, \lambda \exp(-2\pi i<\phi, g(\cdot)> )),$$
 where $<\phi, g(\cdot)>$ denotes the duality pairing of the 1-cochain $\phi$ with the 
 singular 1-simplex $g(\cdot).$ If now $(g, \lambda) \sim (g', \lambda')$, then $\lambda'=\lambda \gamma(g*g'_-)$ by definition. On the other hand,
 in the map $F_n$ the transformed phases $\lambda\mapsto \lambda \exp(-2\pi i  <\phi,g>)$ and 
 $\lambda' \mapsto \lambda' \exp(-2\pi i <\phi, g'>)$ are related by the multiplicative factor
 $\gamma(g*g'_-) \cdot \exp(-2\pi i <\phi, g - g'>) = \gamma(L_n(g)* L_n(g'_-))$ which shows
 that equivalent pairs are mapped to equivalent pairs. That the map $F_n$ is a homomorphism $\hat N_0 \to
 \hat N_0$ follows directly from the definitions. 
 \end{proof}
 
 Since the product in $\hat N_0$ is defined using the composition of paths and the shift in the
 phase factors, $\lambda\mapsto F_n(\lambda),$  is defined by the pairing of the singular
 1-simplices and the cochain $\phi,$ the automorphisms obey the group law $Aut_{n_1n_2} =
 Aut_{n_1} Aut_{n_2}$ for $n_i \in N.$  
 
 For any section $\psi: X \to N$ we can now define a 2-cocycle on $X$ with values in $N_0$ by
 the formula $c(x,x'):=c_{\psi}(x,x'): = \psi(x)\psi(x')\psi(xx')^{-1}$. We recall that in the context of nonabelian groups the cocycle property
 reads
 $$c(x,x') c(xx', x'') = [Aut_{\psi(x)} c(x',x'')] c(x,x'x'')$$ (cf. Equation (\ref{2-cocycle})). We conclude
 
 \begin{proposition} 
 If $c$ can be lifted to a 2-cocycle $
 \hat c$ with values in $\hat N_0$, then $(Aut_{\psi},\hat c)$ is a factor system for the pair $(X,\hat N_0)$. In particular, the corresponding Lie group
 $\hat N = \hat N_0 \times_{(Aut_{\psi},\hat c)} X$ defines an $A$-extension of $N$.
 \end{proposition}

 The cocycle $c$ defined in the proposition is a trivial 2-cocycle, but the potential 2-cocycle extension $\hat c$
 can be nontrivial.  The set of 2-cocycles is in general not a group, except when $A$ is central
 in $\hat N$.
 
The lifting $\hat c$ of $c,$ if it exists, might not be unique; but any other lifting is related to
$\hat c$ by multiplication with a 2-cocycle of $X$ with values in $A;$ see Theorem  \ref{H^2(G,Z)=ext}. The potential obstruction to lifting the cocycle $c$ to a cocycle $\hat c$ with values in $\hat N_0$ is given by the function $\Omega:X^3 \to A$  defined by $$Aut_{\psi(x)}(\hat{c}(y,z)) \hat c(x,yz) = \hat{c}(x, y)\hat{c}(xy,z) \Omega(x,y,z),$$ where $\hat c(x,y)$ is an arbitrary choice of an element in the fibre $\pi^{-1}(c(x,y)).$
 
\begin{remark}
 If we can choose the section $\psi$ as a group homomorphism, then
 $c=1$ and we can take $\hat c=1.$ This occurs for example when $X$ is the free abelian group
 $\Bbb Z.$ One can then choose   $\psi(1)\in \pi^{-1}(1)$ arbitrarily
 and set $\psi(n) = (\psi(1))^n.$ 
\end{remark}
 
 The above discussion can be applied to the case when $G$ is the based loop group $\Omega K$
 of a non simply connected compact group $K$, $N$ is the subgroup of contractible loops and $A$ is the unit circle $S^1$
 endowed with the trivial $G$-action. Now
 $X$ is the fundamental group of $K.$ Take $\hat N$ the central extension of level $k.$ (For
 the definition of the level $k$ see the Appendix A.)
 The corresponding obstruction class $\tau(f) \in \mathrm{H}^3(X, S^1)$ in Theorem \ref{extension} is then represented by $\Omega$ which in this case is a 3-cocycle and can be computed explicitly.  
 
\begin{remark}
 Let $K = SU(n)/\mathbb{Z}_p$, where $p$ divides $n.$ In this case $X= \mathbb{Z}_p$
 is a finite cyclic group and $c(x,y)$  can be taken as a power of a single contractible
 loop $c(1,1)$; for simplicity we may choose $c(1,1)$ equal to the contractible loop given by the homomorphism $t\mapsto
 \exp({it h})\in SU(n)$ with $h=diag(1,1, \dots, 1, -n+1)$ and $\psi(1)$ the non contractible loop
 in $SU(n)/\mathbb{Z}_p$ defined by the path $t\mapsto \exp({it h/p}).$  It follows that the obstruction $\Omega$ arises only from the nontrivial action
 of $\psi(x)$ on $\hat c(y,z).$ This can be evaluated and the result is for the generators
 $x,y,z =1 \in \mathbb{Z}_p$ in a central extension of level $k$ of $\Omega K$
 $$ Aut_{\psi(1)} \hat c(1,1) = e^{ik\pi \cdot n(n-1)/p} \hat c(1,1).$$
 This is equal to $1$ when $n$ is odd or $n/p$ is even; it is equal to $-1$ otherwise. Thus
 $\Omega(1,1,1) = -1$ in the latter case and this represents a nontrivial element in
 $\mathrm{H}^3(\mathbb{Z}_p, S^1).$ This is compatible with Propositions 3.3.1 and  3.6.2 in \cite{TL}: The group $\Omega K$
 is a subgroup of $L_Z K$ in \cite{TL} and the central extension of the latter is
 defined for any multiple of the \emph{fundamental level} $k_f$ which
 is equal to 1 or 2 in the case of $SU(n)/\mathbb{Z}_p.$ 
\end{remark}

If $\Omega$ is a coboundary of a 2-cochain
 $\theta$ then we can replace $\hat c  \mapsto \hat c \theta$ to make $\Omega$ vanish. Thus the
 obstruction in the construction of $\hat N$ is really the class of $\Omega$ modulo coboundaries
 of  2-cochains.

 \section{More on the transgression $\mathrm{H_s^2}(N, A) \to \mathrm{H_s^3}(H, A^N)$: \\ The case $A= Map(G, S^1)$}
 
In the forthcoming sections we shall study locally smooth group cohomology with coefficients
in the module $A$ of smooth circle valued functions on the Lie group $G.$  The $G$-action on
functions on $G$ is taken as the right multiplication on the argument. In that case all the global
cocycles are trivial: In fact, given such a globally defined smooth cocycle $c_n=c_n(g; g_1,\dots, g_n)$  on a Lie group $G$ (that is, $c_n : G^{n+1} \to S^1$ is smooth), then it is
a coboundary of the $n-1$ cochain $d_{n-1}$ defined by
$$d_{n-1}(g; g_1, \dots, g_{n-1}) := c_n(1; g, g_1, \dots , g_{n-1})$$ 
with respect to the usual coboundary operator $\delta$ which in this case reads 
\begin{align*}
(\delta d_{n-1})(g; g_1, \dots, g_n) &= \sum_{k=1}^{n-1} (-1)^{k}d_{n-1}(g; g_1, \dots, g_k g_{k+1}, \dots g_n)\\
&+d_{n-1}(gg_1; g_2, \dots ,g_n) + (-1)^{n} d_{n-1}(g; g_1, \dots , g_{n-1}).
\end{align*}
The equation $c_n = \delta d_{n-1}$ follows directly from the definitions using $\delta c_n =0.$ 
Here we have used the additive notation for the product in the multiplicative abelian group of
circle valued functions. As a consequence, any locally smooth cocycle which can be extended to a globally smooth cocycle
(with coefficients in $Map(G, S^1)$) is trivial. Thus the potential nontriviality is of topological
nature, the obstruction to extending a locally smooth cocycle to a globally smooth cocycle. 
 
We return to the setting of Section 4:  Let $N\subset G$ be a normal subgroup and $H= G/N.$
We continue with an element $\omega_3$ representing a class in the singular cohomology 
(although in the applications later $\omega_3$ is defined as a de Rham form) $\mathrm{H}^3(H, \Bbb Z)$ and the following two assumptions:
 \begin{itemize}
 \item[(1)]
 The pull-back $\pi^*(\omega_3) = d\theta_2$ is trivial on $G$.
 \item[(2)]
 $H$ and $G$ are  2-connected (implying the vanishing of the first and second homology groups).
 \end{itemize}
 Using the exact homotopy sequence from the fibration
 $N\to G \to H$ we conclude that $N$ is 1-connected thus that also
 $\mathrm{H}_1(N,\Bbb Z) =0.$ For each $g\in G$ we select a path $g(t)$ with end points $g(0) = 1\in G$ and $g(1) =g.$ We can make the choice $g\to g(t)$ in a locally smooth manner close to the neutral element
 $1\in G.$ In addition, since also $N$ is connected, we may assume that $g(t) \in N$ if
 $g\in N.$ For a triple $g, g_1, g_2 \in G$ we make a choice of a singular 2-simplex $\Delta(g;g_1, g_2)$ such that its boundary is given by the union of the 1-simplices $gg_1(t), gg_1(1)g_2(t)$
 and $g (g_1 g_2)(1-t).$ All this can be made in a locally smooth manner since locally the Lie groups
 are open contractible sets in a vector space. We set
 $$ c_2(g; g_1, g_2) = \exp{2\pi i <\Delta(g; g_1, g_2), \theta_2>}$$
 using the duality pairing of singular 2-simplices and 2-cochains. This formula does not in
 general define a group cocycle for $G$ but it gives a 2-cocycle for the group $N$ endowed with the 
 right action of $N$ on $G$ and the corresponding action of $N$ on $A = Map(G, S^1).$ 
 To prove that indeed 
 \begin{align*}
&(\delta c_2) (g; n_1, n_2, n_3) 
\\
=&c_2(g; n_1, n_2) c_2(g; n_1n_2, n_3) c_2(g; n_1, n_2 n_3)^{-1} c_2(gn_1; n_2, n_3)^{-1}=1,
 \end{align*}
 we just need to observe that the product is given through pairing the cochain $\theta_2$
 with the singular cycle defined as the union of the singular 2-simplices involved in the
 above formula. All these 2-simplices are in the same $N$ orbit $gN$ and since $d\theta_2
 = \pi^*\omega_3$ the cochain $\theta_2$ is actually an integral  cocycle on the $N$ orbits
 and the pairing gives an integer $k$ and $\exp{2\pi i k}=1.$ For arbitrary $g_i \in G$ the coboundary $\delta c_2$ does not vanish but its value
 $$(\delta c_2)(g; g_1, g_2, g_3) = \exp{2\pi i <\Delta(g; g_1, g_2, g_3), d\theta_2>}$$
 is given by pairing $d\theta_2 = \pi^*\omega_3$ with the singular 3-simplex $V$ with boundary
 consisting of the sum of the faces $\Delta(g; g_1, g_2), \Delta(g; g_1g_2, g_3),\Delta(g; g_1, g_2g_3),$
 \linebreak  $\Delta(gg_1;g_2, g_3).$ 
 But this is the same as $\exp{2\pi i <\pi(V), \omega_3>}$ and therefore it depends only on
 the projections $\pi(g), \pi(g_i) \in H.$ Denote by $c_3= c_3(h; h_1, h_2, h_3)$ this locally
 smooth 3-cocycle on $H.$ (This construction can be extended to higher cocycles under appropriate 
 conditions on the homology groups of $H,$ \cite{WW}.)
 
 \begin{remark}  The transgression above from a 2-cocycle of $N$ to a 3-cocycle of $H$ can be viewed as
 a global version of the descent equations in the de Rham--Cartan--Chevalley--Eilenberg double
 complex, known as the Becchi--Rouet--Stora double complex in quantum field theory, \cite{St}, \cite{Sch}.
 There one studies forms $\omega^{(p,q)}$ of de Rham degree $q$ on a manifold $M$ and of 'ghost degree'
 $p,$ that is, forms of degree $p$ in the Lie algebra cohomology of the Lie algebra of infinitesimal 
 gauge transformations in a vector bundle over $M.$ The forms are local functions of a connnection form
 in the vector bundle.  Starting from a form $\omega^{(0,2n)}$ (for example, the Chern class
 $\text{tr}\, F^n$ where $F$ is the curvature form of a vector bundle) which is closed
 with respect both the de Rham differential $d$ and the Lie algebra coboundary operator $\delta$)
 and writing locally $\omega^{(0,2n)} = d\omega^{(0,2n-1)}$ as an exterior derivative of a Chern-Simons type of form, 
  one
 then generates (using the Cartan homotopy formula) a series of local forms $\omega^{(p,q)}$ such that
 $ \delta \omega^{(p,q)} = d \omega^{(p+1, q-1)}$. In particular, the form $\omega^{(1,2n-2)}$ is the infinitesimal gauge anomaly in dimension $2n-2$
 and $\omega^{(2, 2n-3)}$ is a 2-cocycle  of the algebra of infinitesimal gauge transformations. In the case when
 $\omega^{(0,4)}$ is the second Chern form  the cocycle
 $\omega^{(2,1)}$ defines the central extension of the loop algebra. 
 \end{remark} 
 
 We may think of the cohomology class $[c_3]$ as an obstruction to prolonging the principal
 $N$-bundle $G$ over $H$ to a bundle $\widehat{G}$ with the structure group $\widehat{N}$. Namely, if such
 a prolongation exists then there is a 2-cocycle $c_2$ on $G$ which when restricted to
 $N$ orbits in $G$ is equal to $c_2(g; n_1, n_2).$ If $c_2'$ is another such a 2-cocycle, then 
 $(\delta c'_2)(\delta c_2)^{-1}$ projects to a  a trivial 3-cocycle on $H.$ Conversely, if
 $c_3$ on $H$ is a coboundary of some $\xi_2$, then $c_2' = c_2 (\pi^*\xi)^{-1}$ agrees with
 $c_2$ on the $N$ orbits and so the 
 obstruction depends only on the cohomology class $ [c_3].$ 
 
 Wagemann and Wockel \cite{WW} defined a map from the locally smooth cohomology of a Lie group $H$ 
 to its {\v C}ech cohomology. We point out that there is also a map from the locally smooth group 
 cohomology $\mathrm{H}_s^2(N, A)$ to the twisted {\v C}ech cohomology {\v H}$_{\eta}^2(H,A)$ defined by the formula
 $$c_{ijk}(x) = \hat\eta_{ij}(x)\hat\eta_{jk}(x)\hat\eta_{ki}(x),$$
 with the twisted cocycle property
 $$ c_{ijk}(x) c_{ikl}(x) = (\eta_{ij}(x) \cdot c_{jkl}(x)) c_{ijl}(x)$$
 where  $\psi_i(x)\eta_{ij}(x) = \psi_j(x),$ $\psi_i: U_i \to G$ are local smooth sections
 for an open good cover $\{U_i\}$ of $H$ and the $\hat\eta_{ij}$'s are lifts of the transition
 functions $\eta_{ij}: U_i \cap U_j \to N$ to the extension $\hat N;$ the product in $\hat N$ being
 determined by an element in $\mathrm{H}_s^2(N, A).$ The collection $\{(\eta_{ij}, c_{ijk})\}$ of local  functions on $H$ is just another way to
 describe a gerbe with band $A$ over $H,$ \cite{Mo}. When $A$ is central in $\hat N$ the action
 of $\eta_{ij}$ on the functions $c_{ijk}$ is trivial and {\v H}$_{\eta} = $\v H. Although these twisted {\v C}ech cocycles have values in $A$, they correspond to a cocycle in $\mathrm{H}^3(H, \Bbb Z)$
by the usual way, taking differences of logarithms $\log (c_{ijk})/2\pi\sqrt{-1}$ on intersections $U_{ijkl}$ 
which must be integer constants for a good cover. 

\section{The case of a non simply connected base $H$}

\subsection{The case of a torus} 

The geometric construction of the locally smooth 3-cocycle on $H$ from the abelian extension
$\hat N$ does not work when $N$ is disconnected, or when $H$ is not simply connected. A simple example
is the case when $H$ is a torus $T^n$, $G$ is the abelian group $\Bbb R^n,$ and $N=\Bbb Z^n.$
 The aim of this section is to explain the relations between
 \begin{enumerate}
 \item
 gerbes on tori, classified topologically by $\mathrm{H}^3(T^n, \mathbb{Z})$,
 \item
 central extensions of the transformation groupoid $\mathbb{R}^n \times \mathbb{Z}^n
\to \mathbb{R}^n,$
 \item
  and the locally smooth group cohomology $\mathrm{H_s^3}(T^n, S^1).$ 
\end{enumerate}  

\noindent
The relation between (1) and (2) follows easily from the definitions.  A gerbe on $T^n$ is a projective complex vector
bundle over $T^n.$ 
The equivalence classes of these projective vector bundles is classified by the Dixmier-Douady
class, an element in $\mathrm{H}^3(T^n, \mathbb{Z}).$ 

In the nontorsion case of
the Dixmier-Douady class we have to use Hilbert bundles. The structure group
of a Hilbert bundle is the unitary group $U(E)$ of a Hilbert space $E.$ 
This is contractible both for the strong topology and the norm topology.
It depends on the concrete construction which topology is appropriate.
In the case involving representations of infinite dimensional groups like
the loop group one has to choose the strong topology since the nontrivial
projective loop group representations are not continuous in the norm topology.

The pull-back of a projective bundle with respect to the covering
map $\mathbb{R}^n \to T^n$ becomes trivial on $\mathbb{R}^n$ and is a projectivization $PE=E/S^1$ of a trivial vector bundle $E=\mathbb{R}^n \times\mathcal{H}$ with $\mathcal{H}$ an infinite dimensional complex Hilbert space.
Since $PE$ is a pull-back of a projective bundle over $T^n$ the fibers over the points $x$ and $x+z$ for $x\in\mathbb{R}^n$ and $z\in\mathbb{Z}^n$ can be identified by an element $g(x,z)
\in PU(\mathcal H)$ with $g(x, z+w) = g(x+z, w) g(x,z).$ Choosing a lift ${\hat g}(x,z) \in U(\mathcal H)$ we have
$${\hat g}(x, z+w) = {\hat g}(x+z,w) {\hat g}(x,z) c(x; z,w),$$
where $c(x; z,w) \in S^1$ satisfies the cocycle condition
$$c(x; z,v) c(x; z+v,w)= c(x; z, v+w) c(x+z;v,w).$$
  
The cocycle $c$ is trivial if and only if $c(x;z, w) =  f(x; z+w) f(x; z)^{-1} f(x; w)^{-1},$ i.e.,
the modified unitaries $g'(x; w) = g(x,w) f(x;w)$ define a cocycle with values in the unitary group
$U(\mathcal H)$ and the projective bundle with fiber $\mathcal{H}/S^1$ over $T^n$ comes form a  vector bundle. 

Alternatively, the transformation groupoid cocycle $c$ can be viewed as an $A$ valued 2-cocycle
on $\mathbb{Z}^n,$ with $A=C^{\infty}(\mathbb{R}^n, S^1),$ see the discussion in the second paragraph of Section 4. In this case $A$ is actually an infinite dimensional Lie group with a Fr\'{e}chet topology.

Nontrivial 2-cocycles $c$ can be easily defined. Indeed, for any totally antisymmetric collection $S_{pqr}$ of integers, with $p,q,r\in \{1,2,\dots,n\}$,
the cochain $c_S(x; u,v)= \exp{(2\pi i \sum S_{pqr} x_p u_q v_r)}$ is easily seen to be a 2-cocycle, with
$x\in \Bbb R^n$ and $u,v \in \Bbb Z^n.$ The rank of the $\Bbb Z$ module of the tensors $S$ is equal to
$\frac{n!}{3! (n-3)!}$ which is equal to the rank of $\mathrm{H}^3(T^n, \Bbb Z).$ 
To check that the cocycles $c_S$ span the group cohomology, $\mathrm{H}^2(\Bbb Z^n, A),$ as a
$\mathbb{Z}$ module,(which is isomorphic to the group $\mathrm{H}^3(T^n, \mathbb{Z}))$ one still needs to check that they are independent as cohomology classes.
For that we use the homomorphism $\mathrm{H}^2(\mathbb{Z}^n, A) \to \mathrm{H}^3(\mathbb{Z}^n,
\mathbb{Z})$ defined as $\frac{1}{2\pi i}\delta \log c_S$ which is the integral 3-cocycle $c'=\sum S_{pqr} w_p u_q v_r.$ Finally, by a theorem due to David Wigner \cite{Wi}, the group cohomology $\mathrm{H}^3(\mathbb{Z}^n, \mathbb{Z})$ is isomorphic to the singular cohomology $\mathrm{H}^3(T^n, \mathbb{Z})$ and in this case the correspondence is simply $c' \mapsto \sum S_{pqr} d\phi_p
\wedge d\phi_q \wedge d\phi_r,$  the singular cocycles represented as de Rham forms. This
shows that the homomorphism is onto and since the ranks (as free $\mathbb{Z}$ modules) match, since gerbes over $T^n$
are classified by $\mathrm{H}^3(T^n,\mathbb{Z}) \simeq \mathrm{H}^2(\mathbb{Z}^n, A)$, the homomorphism is
actually an isomorphism.

However, although
the obstruction $\Omega$ can be defined also in this case, it is not a 3-cocycle. The results of  Theorem \ref{Gammasmoothcrossedmodule} and Theorem \ref{extension} 
do not apply in this case since the first cohomology $\mathrm{H}^1(\Bbb Z^n, A)$ does not vanish.
This cohomology classifies the circle bundles over $\Bbb R^n/\Bbb Z^n =T^n$ and its dimension is the
rank of $\mathrm{H}^2(T^n, \Bbb Z)$ which is equal to $\frac{n!}{2 \cdot (n-2)!}.$ 

For the commutative group $T^n$ the Lie algebra cohomology can be identified as its de Rham cohomology
and is nontrivial for each $n\geq 3.$ But the non simply connectedness of $T^n$ is an obstruction to
lifting the Lie algebra cocycles  to global, locally smooth group cocycles. 

\subsection{A gauge theory application}

The obstruction in the construction of locally smooth 3-cocycles appear also in
gauge theory applications. Let $\mathcal{A}$ be the (Fr\'{e}chet) space of smooth 1-forms on the unit sphere $S^n$ with values in a Lie
algebra $\mathfrak k$ of a compact  Lie group $K.$  Let $\mathcal{K}$ be the gauge group
of all smooth based maps from $S^n$ to $K$ and $\mathcal A/\mathcal K$ the moduli space of gauge connections.
So when $n=1$ we have $\mathcal A/\mathcal K = K$ but for $n > 1$  the moduli space is an infinite-dimensional
Frechet manifold and can be identified (up to homotopy) as $H=Map(S^{n-1}, K)$ using parallel transports along half circles
from the South pole to the North pole on $S^n.$ In this realization we may replace $\mathcal A$ by
the group $G=PH$ of smooth paths $f$ on $H$ starting from the identity and the projection $\pi: \mathcal A
\to H$ is then $\pi(f) = f(1),$ the path evaluated at the end point of the interval $0\leq t\leq 1.$ 
Here $N=\mathcal{K}$ is the based loop group of $H$ with $G/N =H.$ However, $H$ is typically non simply connected and
$N$ disconnected. For example, when $K= SU(p)$ with $p\geq n\geq 3$ then $\pi_1(H) = \pi_{n}(K) = \Bbb Z$ for
odd $n>1$ 
and the connected components of $N$ are labelled by $\Bbb Z.$ In this case $\pi_1(N) = \pi_{n+1}(K) =0$ (for odd $n$)
and $\pi_2(N) =\Bbb Z$ in each connected component of $N$ and so $\mathrm{H}^2(N, \Bbb Z) = \Bbb Z$
(in each connected component) leading to nontrivial extensions of the group of based maps $N = Map_0(S^n, K)$ by the module
$A$ of $S^1$-valued functions on $G$. 

However, in order that $A$ is a locally convex Lie group, we have to restrict to a certain
submodule of the module of all smooth functions $G\to S^1,$ see the discussion
in the second paragraph of Section 5.

These in turn correspond to gerbes over the group $H,$ the moduli space of gauge
connections on $S^n.$ When $n$ is odd and $K= SU(p)$ with $p$ large compared to $n$ by the Bott periodicity the cohomology of $H=Map_0(S^{n-1},K)$ in low degrees is generated by
elements $\alpha_{2k+1}$ in odd degrees ($k=1,2, \dots$). In particular,  $\mathrm{H}^2(H, \Bbb Z)$ vanishes 
for $p > n$ and there are no nontrivial
circle bundles over $H.$ On the other hand, an element $c_1\in \mathrm{H}^1(N, A)$ describes
a circle bundle $Q$ over $G/N$ in the following way: Elements in $Q$ are equivalence classes of pairs $(g, \lambda) \in G \times S^1$
with the equivalence relation $(g, \lambda) \sim (gu, c_1(g;u)\lambda).$ The bundle can be trivialized if and only if
$c_1(g; u)= f(gu)f(g)^{-1}$ for some $f: G \to S^1.$  Thus $\mathrm{H}^1(N, A) =0$ and the transgression map  in Theorem \ref{Gammasmoothcrossedmodule} is defined, provided that $A$ has a Lie
group structure, see the second paragraph of this subsection.

When $n=1$, we are in the situation of Theorem \ref{main theorem on crossed modules vs. loop groups}. In this case we can take the abelian group $A$ as the group $S^1$ (i.e., constant
$S^1$ valued maps on the path group $G$) and $N$ is the based loop group $\Omega K,$ $K=H$ and the transgression is actually an isomorphism. 

If $n$ is even, then $\mathrm{H}^2(H, \Bbb Z)=\Bbb Z$ for $K$ a simple compact Lie group (again for $p>n$)
and the transgression
map in Theorem \ref{Gammasmoothcrossedmodule} is ill-defined. On the other hand, now $\mathrm{H}^3(H, \Bbb Z)=0$ and there are no
nontrivial gerbes over $H$ which is reflected in the fact that there are no topologically nontrivial abelian extensions
of $N=Map_0(S^n, K)$ by $A.$ The latter property follows again from Bott periodicity:
for large $p$ the group $N=Map_0(S^{n}, K)$ has in the low dimensions the cohomology of the infinite
unitary group $U(\infty)$ which is generated by one element in each odd degree.
In particular, $\mathrm{H}^2(N,\Bbb Z) =0.$ On the other hand, the extension
of on $N$ by $A=Map( G , S^1)$ is the same as the circle extension of the action
groupoid $G\times N \to N.$ Since $G$ is contractible nontriviality of the
circle extension would mean that there is a topologically nontrivial circle bundle
over $N$ which is impossible by $\mathrm{H}^2(N, \Bbb Z) =0.$

\appendix \section{The central extension of the\\ smooth based loop group $\Omega G$}\label{central ext based lg}

The aim of this Appendix is to extend the geometric construction of the central extension $\widehat{LG}$ of the
smooth loop group $LG,$ with $G$ a simply connected compact Lie group, \cite{Mi1}, to the case when $G$ is a connected compact Lie group but not necessarily simply connected. Since for the
construction of (non equivariant) gerbes over $G$ one needs only the based loop group
$\Omega G$ we restrict later to this case; there is an extra twist in the construction for
$LG$ when $G$ is not simply connected. The reason is that in $\Omega G$ each connected
component is simply connected but not in $LG.$ The construction in this section is on the level of abstract groups, independently of
possible (highest weight) representations.
Let us first briefly recall the construction when $G$ is simply connected, \cite{Mi1}:

Let $D$ denote the unit disc in $\Bbb C$ including the boundary circle $S^1$.
Let $G$ be a compact simply connected simple Lie group, and denote by $DG$ the space of all
smooth mappings $f:D\to G$ such that the radial derivatives of $f$ approach zero
(to  all orders) at the boundary and at the center of the disk. 
The construction is easily extended to the case of a semisimple Lie
group because of the factorization to simple factors.

Let $\mathcal{G}$ be the subgroup of $DG$ consisting of maps $f:D\to G$ such that $f
=1$ on the boundary (and all its radial derivatives vanish on the boundary.) Clearly $\mathcal{G}\subset DG$ is a normal subgroup. Define
$$\gamma(f_1,f_2)=\frac{\theta^2}{16\pi^2}\int_{D} <f_1^{-1}df_1,df_2 f_2^{-1}>,
$$
where $\theta$ is the length of the longest root of the Lie algebra $\mathfrak{g}$,
computed using the dual of the invariant inner product $<\cdot,\cdot>$ in $\mathfrak{g}.$
Is is easily shown by a direct computation that  $\gamma$ is a real valued
2-cocycle on $DG$,
$$\gamma(f_1,f_2)+\gamma(f_1f_2,f_3)=\gamma(f_2,f_3)+\gamma(f_1,f_2f_3).$$
From the cocycle property it follows that we can define an extension of the
group $DG$ by $S^1$ using the multiplication rule
$$(f_1,\lambda_1)(f_2,\lambda_2)=(f_1f_2,\lambda_1\lambda_2 \exp[2\pi i\gamma(f_1,f_2)]).$$
Any $g\in\mathcal{G}$ can be thought of as a map $g:S^2\to G$ by identifying the boundary
$S^1$ of $D$ with the north pole of $S^2$.
There is a homomorphism $\phi:\mathcal{G}\to DG\times S^1$ defined by
$$\phi(g)=(g,\exp[2\pi iC(g)]),$$
where
$$\aligned C(g)&=\frac{\theta^2}{48\pi^2}\int_{B}<dgg^{-1},\frac12[dgg^{-1},dgg^{-1}]>
\\&=\frac{\theta^2}{48\pi^2}\int_{B} \epsilon^{ijk}<g^{-1}\partial_i g,\frac12[
g^{-1}\partial_j g,g^{-1}\partial_k g]>,\endaligned$$
and we have extended the map $g:S^2\to G$ to a map $g:B\to G$, where $B$ is the
unit ball in $\Bbb R^3$ with boundary $S^2$. Owing to the arbitrariness in the choice
of the extension $C(g)$ is not uniquely defined. Let $g,g'$ be two extensions to
$B$. Because they agree on the boundary we can glue them along the boundary to
form a mapping $h:S^3\to G$. The upper hemisphere of $S^3$ corresponds to one copy
of $B$ and the lower hemisphere to a second copy. It is also known that
$$\frac{\theta^2}{48\pi^2}\int_{S^3}<dhh^{-1},\frac12[dhh^{-1},dhh^{-1}]>
$$
is an integer, see section 4.2 in \cite{Mi3}.  
We  conclude that the difference $C(g)-C(g')\in \Bbb Z$  and
$\exp[2\pi iC(g)]$ is well-defined. The condition $\phi(g_1g_2)=\phi(g_1)\phi(g_2)$ is equivalent to
$$C(g_1g_2)=C(g_1)+C(g_2)+\gamma(g_1,g_2)\,\,mod\,\Bbb Z$$
which can be verified by a direct substitution. Finally, \cite{Mi1},

\begin{proposition}  The image $\phi(\mathcal{G})$ in $DG\times S^1$ is a normal
subgroup and the quotient $(DG \times S^1)/\phi(\mathcal{G})$ is the basic central extension of $LG.$
\end{proposition}

So far we have not said anything about the differentiable structure of $LG$
or $\widehat{LG}$. We shall use in $LG$ the topology of uniform convergence of maps
$f:S^1\to G$ and of all their partial derivatives. Let $U$ be an open neighborhood
of 0 in $\mathfrak{g}$ such that the exponential mapping $\mathfrak{g}\to G$ is one-to-one on $U$.
An open neighborhood $U_f$
of an arbitrary point $f\in LG$ consists of all maps of the form $e^{Z}f$, where
$Z$ is any smooth map from $S^1$ to $U$. The mapping $\xi_f:e^{Z}f\mapsto Z$
defines a coordinate system on $U_f$.The topology of the smooth loop algebra $L
\mathfrak{g}$ is defined by the infinite system of seminorms
$$ \Vert Z\Vert_{(n)}=
\sqrt{\int_{S^1}\biggl<\frac{d^n}{d\varphi^n}Z,\frac{d^n}{d\varphi^n}
Z\biggr > }.$$
The coordinate transformation $\xi_f\circ\xi_{f'}^{-1}(Z)=\log(e^Zff'^{-1})$ is
differentiable in the Fr\'echet sense and by definition $LG$ is a \it Fr\'echet
manifold. \rm Starting from the basic central extension, for any integer $k$ we can construct a central extension of $LG$
of level $k$  by replacing $\gamma(f_1,f_2)$ by $k\gamma(f_1,f_2)$ and $C(g)$ by $k C(g)$.

Next we relax the requirement of simply connectedness on $G$ but concentrate on $\Omega G.$
For any element $a\in \pi_1(G)$
choose a based loop $h_a$ in $G$ representing $a.$ We may fix $h_e =1,$ the constant loop at
the identity element in $G$ when $e$ is the identity element in $\pi_1(G).$ 
Elements in the central extension $\widehat{LG}$ are now represented by pairs $(f,\lambda)$,
where $\lambda\in S^1$ and $f:[0,1] \times S^1 \to G$ is a smooth map such that $f(0,\cdot)$
is the loop $h_a$ when $f(1,\cdot)$ has the homotopy type of $h_a,$  and all the derivatives of
$f$ with respect to the first argument vanish at $t=0,1.$

We define again an equivalence relation $(f,\lambda) \sim (f',\lambda')$ when $f'= fg$
and $\lambda'=\lambda e^{2\pi i C(g)}$, where $g(t, \cdot) =1$ for $t=0,1$ and $C(g)$ is defined using any smooth extension of $g$ to the solid torus with boundary $S^1\times S^1$. Here, the first circle is the interval $[0,1]$ with endpoints identified. The normalization of 
$$C= \int_B \omega_3$$ 
comes from the requirement that the pairing  of the closed 3-form $\omega_3$ with any 3-cycle is
an integer, that is, $\omega_3 \in \mathrm{H}^3(G, \Bbb Z).$ The third cohomology is equal to $\Bbb Z$ 
for any simple compact Lie group. Writing $\omega_3$  in the form
$$\omega_3 =  \frac{k}{24 \pi^2}  <dg g^{-1}, [dg^{-1}, dg g^{-1}]>$$
with $<\cdot, \cdot>$ the invariant bilinear form on the Lie algebra of $G$ normalized
such that the length squared of the longest root is equal to 2, the level $k$ has to be
a multiple of the \emph{fundamental level} $k_f$, \cite{FGK, TL}.

The above construction defines a circle bundle over $\Omega G.$ The product is defined as follows.
If $f_a$ has the homotopy type $a$ and $f_b$ the homotopy type $b$ we choose $c$ such that
$h_c$ represents the element $ab$ in $\pi_1(G).$ We choose a homotopy $h(a,b)$ joining 
$h_ah_b$ to $h_c.$ Together with $f_af_b$ we have a homotopy $f_{a,b}$ joining the loop
$(f_af_b)(1,\cdot)$ to the loop $h_c.$ Then we can define
$$ (f_a, \lambda)\cdot (f_b, \mu) = (f_{a,b}, \lambda\mu e^{2\pi i\gamma(f_a, f_b)})$$
and this descends again to the equivalence classes represented by the above pairs. However, for a construction
of a highest weight representation of the full loop group $LG$ more is needed. 
The level has to be a multiple of the \emph{basic level} $k_b$ which in turn is a multiple
of $k_f;$ these levels have been listed in \cite{TL}. Each multiple of the fundamental 
level produces representations of the based loop group $\Omega G.$ The group $G$ of constant
loops acts on the central extension of $\Omega G$ as a group of automorphisms but these
automorphisms can act in the representation space only for certain highest weights.

\begin{example}
For $G= SO(3)$ we have $\pi_1(G) = \Bbb Z_2$ and we have only one nontrivial $h$
with $h^2$ homotopic to the trivial loop. The pull-back of the basic third cohomology class in
the projection $SU(2) \to SO(3)$ is twice the basic class in $\mathrm{H}^3(SU(2), \Bbb Z).$ 
The basic class in $\mathrm{H}^3(SU(2), \Bbb Z)$ corresponds to the level $k=1$ central extension,
so the basic class in $\mathrm{H}^3(SO(3), \Bbb Z)$ defines the level $k=2$ central extension;
the Lie algebras of the central extension of the loop groups $L SU(2)$ and $L SO(3)$
are isomorphic, so the counting of the algebraic level on the Lie algebra is the same in both cases. In this case the basic level is the same as the fundamental level $k_f =2.$ 
\end{example}

 \section{A Fock space construction of gerbes over $G/Z$ } 
 
 Let $G$ be a simply connected compact Lie group and $Z$ a cyclic subgroup of its center.
 Thus the fundamental group of $G/Z$ is $Z.$ Fix an irreducible representation of $G$
 in a complex vector space $\Bbb C^N.$ Then the loop group $LG$ acts by pointwise
 multiplication in the Hilbert space $\mathcal{H}= L^2(S^1, \Bbb C^N).$ Likewise, identifying the
 based loop group $\Omega G$ as the group of contractible based loops $\Omega_0(G/Z)$ in $G/Z$ 
 we have a representation of the latter in $\mathcal{H}.$ The loop group $\Omega (G/Z)$ is identified 
 as the group $\Omega_Z G$ of paths in $G$ starting from the identity in $G$ and with the other
 end point at $Z\subset G.$ This group splits to components $\Omega_z G$ labelled by $z\in Z.$
 The center $Z$ acts in $\Bbb C^N$ as multiplication by roots of unity.
 
 A gerbe over a compact simple Lie group comes always from an infinite dimensional projective
 unitary representation
 of the based loop group $\Omega G.$ This follows from the fact that we have the principal
 fibration $PG\to  PG/\Omega G = G$ and the pull-back gerbe on
 $PG$ is trivial since that group is contractible.  So a gerbe over $G$ comes
 from a homomorphism from $\Omega G$ to the projective unitary group of a complex Hilbert
 space.  The remaining question is what is the relation of the Dixmier-Douady class 
 of the gerbe, an element in $\mathrm{H}^3(G, \Bbb Z) = \Bbb Z$ to the level $k\in \Bbb Z$ of
 the projective representation of $\Omega G.$ These are equal when $G$ is simply connected
 but they might differ when $G$ is not simply connected, see the Example A.2 in the end of the last section.  
 
 The standard Fock space construction of representations of a central extension of the loop
 group uses the embedding $LG \subset U_{res}(\mathcal{H}_+ \oplus \mathcal{H}_-)$, where 
 $\mathcal{H}=\mathcal{H}_+ \oplus \mathcal{H}_-$
 is the polarization of $\mathcal{H}$ to non negative and negative Fourier modes, \cite{PrSe86}.
 Let us denote by $\epsilon$ the sign operator in $\mathcal{H},$ restricted to $\mathcal{H}_{\pm}$
 it gives the value $\pm 1.$ 
 The restricted unitary group $U_{res}$ consists of unitaries such that the off-diagonal 
 blocks are Hilbert-Schmidt operators.  According to \cite{SS} a unitary operator $g$ in the
 "one particle space" $\mathcal{H}$ can be lifted to a unitary operator in the Fock space $\mathcal{F}$
 with a vacuum corresponding to the energy polarization $\mathcal{H}=\mathcal{H}_+\oplus \mathcal{H}_-$ if and only if 
 $g$ belongs to $U_{res}.$ By a simple Fourier analysis, this is the case when $g$ is an 
 element of $LG$ acting as a multiplication operator in $\mathcal{H}.$ However, it is not difficult
 to see that the Hilbert-Schmidt condition is not satisfied when $g\in \Omega_z G$ when $z$ is
 not the neutral element in $Z.$ 
 
 On the level of Lie algebras, the operators $X,Y$ which satisfy the Hilbert-Schmidt condition above,
 lead to a central extension of the Lie algebra in the Fock space described by the 2-cocycle \cite{Lu}
 $$c(X,Y) = \frac14 \text{tr}\, \epsilon[\epsilon, X], [\epsilon,Y],$$
 which when restricted to the loop algebra gives its standard central extension \cite{PrSe86}.
 
 Let next $\mathcal{H}_Z= \oplus \mathcal{H}_z$, where $\mathcal{H}_z$ is the space $C^{\infty}_z([0,2\pi],\Bbb C^N)$ of
 smooth functions $\psi$ with the boundary condition $\psi(2\pi) = z\cdot \psi(0).$ The
 smooth subspace is in the domain of $D= -i \frac{d}{d\phi}$ and we can define a polarization
 $\mathcal{H}_z = \mathcal{H}_{z+} \oplus \mathcal{H}_{z-}$ to non negative and negative modes for the operator $D.$ 
 The based loop group $\Omega G$ acts unitarily in each $\mathcal{H}_z.$ The elements of $\Omega_{z'}G$
 map $\mathcal{H}_z$ onto $\mathcal{H}_{zz'}.$ 
 
  \begin{lemma} For each $f\in \Omega_z G$ the operator $[\epsilon, m(f)]$ is Hilbert-Schmidt, where 
 $m(f)$ is the  multiplication operator by $f$  in $\mathcal{H}.$  \end{lemma}
 \begin{proof} The eigenvectors of the free Dirac operator $D$ in $\mathcal{H}_0$ are the Fourier modes 
 $e^{2\pi i nx} v$ with $v\in \Bbb C^N$ and $n\in \Bbb Z.$ We select a basis in $\Bbb C^N$ consisting of eigenvectors $v$ of $Z.$ The eigenvectors of $D$  in $\mathcal{H}_z$ are the functions $e^{2\pi i (n+\frac{1}{p}) x}v$, where $e^{2\pi i/p}$ is the eigenvalue of $z$ corresponding to the eigenvector $v$
 of $Z.$ (Since $G$ is compact simply connected the order of each element in $Z$ is finite.)
 The off-diagonal blocks with respect to the polarization by $\epsilon$ of $m(f)$ for $f\in \Omega G$ 
are estimated as in the proof of Proposition 6.3.1 in \cite{PrSe86} , their Hilbert-Schmidt norm is proportional to $\sum |n| |a_n|^2$,
where $f(x)= \sum f_n(x) =  \sum a_n e^{2\pi i nx}$ with constant $N\times N$ matrices $a_n$. This norm is finite even for $1/2$ differentiable
functions (in the Sobolev sense).  The multiplication operator $m(z)$ maps isometrically the
positive Fourier modes in $\mathcal{H}_{z'}$ to the Fourier  in $\mathcal{H}_{zz'}.$ 
It follows that $m(z)$ trivially satisfies the Hilbert-Schmidt condition on off-diagonal
blocks. \end{proof} 

For a simple simply connected compact Lie group $G$ its center is cyclic except in the cases
$G = D_{2\ell}$ for $\ell=2,3, \dots$  when the center is $\mathbb{Z}_2 \times \mathbb{Z}_2, $  see e.g. \cite[Theorem 15.23.]{SW}. An immediate conclusion of the above Lemma is the following statement:

\begin{proposition} Let $G$ be a compact simply connected Lie group and $Z$ a cyclic subgroup of the
center of order  $p.$ The Fock space construction of the representations of the 
central extension of $\Omega(G/Z)$ gives a representation of level $pk$, where $k$ is the level
of the representation of $\Omega_0 G$ arising from the canonical quantization in the sector $\mathcal{H}_0.$
The subspace of $Z$ invariant vectors in $\mathcal{F}$ carries then a representation of the full loop group $L(G/Z).$ Since $Z$ is cyclic its action in $\Bbb C^N$ is given by 
powers of a complex phase $a$ with $a^p =1.$ Here the action of $Z$ in $\mathcal{F}$ is defined by fixing the action of the generator $z$ of the center on the Fock vacuum to be trivial and the action on $n$ particle states as
a multiplication by $a^n.$ 
 \end{proposition}
 
\begin{remark}
The minimal level of highest weight representations of $L(G/Z),$ called
the basic level and denoted by $k_b$ in \cite{TL} can be lower than $p.$ For example, in the
case of $G=SU(N)$ and $Z=\mathbb{Z}_r=\mathbb{Z}/r\mathbb{Z}$ for some $r$ dividing $N$ the basic level is $p=r$ if $r$ does
not divide $N/r,$ otherwise it is $N/r^2.$ In particular, for $Z=\mathbb{Z}_N$ the basic level is 
always $N.$ The basic levels for all cases are listed in \cite[Proposition 3.6.2]{TL}.
\end{remark}

\begin{remark}
In the case of the based loop group $\Omega(G/Z)$ the minimal level (called
the fundamental level and denoted by $k_f$ is either 1 or 2).  The reason is that in this case
there is a less stringent condition for the highest weights, they do not need to define
representations of $G/Z\subset L(G/Z).$ The representations are certain direct sums of
irreducible highest weights representations of $LG$ with an action of $Z$  permuting
the irreducible components.  When $G$ is simple the fundamental level depends only on the image 
$\pi(1)$ of the generator  $1\in \mathrm{H}_3(G)= \Bbb Z$ in $\mathrm{H}_3(G/Z).$ The element $\pi(1)$ is $k_f$
times a generator of $\mathrm{H}_3(G/Z)$ as follows from the homological argument in the
Appendix 1  \cite{FGK}, compatible with the calculations in the Section 6 in the same reference
and the listings in Proposition 3.6.2 in \cite{TL}.
\end{remark}

 \section{Extensions of smooth crossed modules}\label{section extension smooth crossed module}

Let $(\alpha,\widehat{S})$ be smooth crossed module (cf. Definition \ref{crossed modules}). In this section we show that the cohomology group $\text{H}^2_{ss}(G/N,Z)_T$ ($=\text{H}^2_{s}(G/N,Z)_T$ if $G/N$ is connected) operates in a natural way as a simply transitive transformation group on the set of equivalence classes of extensions of the smooth crossed module $(\alpha,\widehat{S})$ (cf. \cite[Theorem III.8]{Ne06}). For the sake of a concise presentation we recall the following notation: Given a map $\sigma:G\rightarrow N$ between two groups $G$ and $N$, we write $\delta_\sigma:G\times G\rightarrow N$ for the induced map defined by $$\delta_\sigma(g,g'):=\sigma(g)\sigma(g')\sigma(gg')^{-1}.$$ Moreover, we recall that $\sigma\in C^1_s(G,N)$ implies $\delta_\sigma\in C^2_{ss}(G,N)$ (cf. Section \ref{pre and not}).

\begin{definition}~\label{continuation of lie groups}
\begin{enumerate}
\item
Let $\alpha:\widehat{N}\rightarrow G$ be a morphism of Lie groups such that $N:=\im(\alpha)$ is a split normal subgroup of $G$. Further, let $q:G\rightarrow G/N$ be the corresponding quotient map. A $\emph{continuation}$ of $\alpha:\widehat{N}\rightarrow G$ is a pair $(\widehat{\alpha},\widehat{G})$, consisting of a Lie group extension $\widehat{G}$ of $G/N$ by $\widehat{N}$ and a morphism  $\widehat{\alpha}:\widehat{G}\rightarrow G$ of Lie groups such that the following diagram commutes
\[\xymatrix{ 1 \ar[r] & \widehat{N} \ar[r] \ar[d]^{\alpha} & \widehat{G}\ar[r]\ar[d]^{\widehat{\alpha}}& G/N\ar[r] \ar@{=}[d] & 1\\
1\ar[r]& N\ar[r] & G\ar[r]^{q} & G/N \ar[r] & 1.}
\]It is not hard to check that $\ker(\widehat{\alpha})=Z$ and that $\widehat{\alpha}$ is surjective.
\item
Two continuations $(\widehat{G}_1,\widehat{\alpha}_1)$ and $(\widehat{G}_2,\widehat{\alpha}_2)$ of $\alpha:\widehat{N}\rightarrow G$ are called \emph{equivalent} if there exists a Lie group isomorphism $\varphi:\widehat{G}_2\rightarrow\widehat{G}_1$ satisfying
\begin{align}
\varphi_{\mid_{\widehat{N}}}=\id_{\widehat{N}} \quad \text{and} \quad \widehat{\alpha}_2=\widehat{\alpha}_1\circ\varphi.\notag
\end{align} 
\item
An $\emph{extension}$ of a smooth crossed module $(\alpha,\widehat{S})$ is a continuation $(\widehat{G},\widehat{\alpha})$ of $\alpha:\widehat{N}\rightarrow \widehat{N}$ satisfying $\widehat{S}\circ\widehat{\alpha}=C_{\widehat{G}}:\widehat{G}\rightarrow\Aut(G)$.
\end{enumerate}
\end{definition}

\begin{remark}\label{remark on equivalence}
Note that if $(\widehat{G}_1,\widehat{\alpha}_1)$ and $(\widehat{G}_2,\widehat{\alpha}_2)$ are equivalent extensions of a smooth crossed module $(\alpha,\widehat{S})$, then $\widehat{G}_1$ and $\widehat{G}_2$ are also equivalent as Lie group extensions.
\end{remark}


We continue with a Lie group extension $q:\widehat{G}\rightarrow G$ of the Lie group $G$ by the Lie group $N$ and a morphism $\alpha:N\rightarrow H$ of Lie groups. Furthermore, we assume that $\ker(\alpha)$ is a split Lie subgroup of $N$ and that $\im(\alpha)$ is a split Lie subgroup of $H$ for which $\alpha$ induces an isomorphism $N/\ker(\alpha)\rightarrow\im(\alpha)$. In this case a few moments thought shows that the map
\begin{align*}
\Phi_{\alpha}:\Aut(N,\ker(\alpha))\rightarrow\Aut(\im(\alpha)), \quad \Phi_{\alpha}(\varphi)(\alpha(n)):=\alpha(\varphi(n)),
\end{align*}
where the set on the left-hand side is the group of all Lie group automorphisms of $N$ preserving the split Lie subgroup $\ker(\alpha)$, is a well-defined group homomorphism.


\begin{lemma}\label{technical lemma 2}
Let $\mu:G\rightarrow\widehat{G}$ be a normalized locally smooth section of $q$ and $\nu\in C^1_s(G,H)$. Furthermore, let $C_{N}$ denote the conjugation action of $\widehat{G}$ restricted to the normal subgroup $N$. Then the following statements are equivalent:
\begin{itemize}
\item[(a)]
There exists a unique morphism of Lie groups $\widehat{\alpha}:\widehat{G}\rightarrow H$ satisfying the equations $\widehat{\alpha}_{\mid_N}=\alpha$ and $\widehat{\alpha}\circ\mu=\nu$.
\item[(b)]
The group $\im(\alpha)$ is a normal subgroup of the group generated by itself and $\nu(G)$, and we have $\alpha\circ\delta_{\mu}=\delta_{\nu}$ and $\Phi_{\alpha}(C_N\circ\mu)=C_{\emph{\im}(\alpha)}\circ\nu$. 
\end{itemize}
\end{lemma}
\begin{proof}
We first recall that every element of $\widehat{G}$ can be uniquely written as $n\mu(g)$ for $n\in N$ and $g\in G$. Then we claim that the map
\[\widehat{\alpha}:\widehat{G}\rightarrow H, \quad \widehat{\alpha}(n\mu(g)):=\alpha(n)\cdot\nu(g)
\]satisfies the requirements of part (a). In fact, 
\begin{align}
\widehat{\alpha}((n\mu(g))\cdot(n'\mu(g')))
&=\widehat{\alpha}(n\cdot (C_N(\mu(g)))n'\cdot\delta_{\mu}(g,g')\cdot\mu(gg'))\notag\\
&=\alpha(n)\cdot(\Phi_{\alpha}(C_N(\mu(g)))(\alpha(n'))\cdot\alpha(\delta_{\mu}(g,g'))\cdot\nu(gg')\notag\\
&=\alpha(n)\cdot( C_{\im(\alpha)}\circ\nu(g))(\alpha(n'))\cdot\delta_{\nu}(g,g')\cdot\nu(gg')\notag\\
&=\alpha(n)\cdot\nu(g)\cdot\alpha(n')\cdot\nu(g')=\widehat{\alpha}(n\mu(g))\cdot\widehat{\alpha}(n'\mu(g')),\notag
\end{align}
shows that $\widehat{\alpha}$ is a homomorphism. That it is actually a morphism of Lie groups follows from the fact that it is smooth in an identity neighbourhood.
The other direction follows from a simple calculation.
\end{proof}

\begin{definition}\label{smooth structural cocycle}
Let $(\alpha,\widehat{S})$ be a smooth crossed module. The elements of the set
\[Z^2_{ss}(\alpha,\widehat{S}):=\{(f,\sigma)\in C^2_{ss}(G/N,\widehat{N})\times C^1_s(G/N,G):\alpha\circ f=\delta_{\sigma},\,d_{\widehat{S}\circ\sigma}f=1_{\widehat{N}}\}
\]are called \emph{smooth structural cocycles} of $(\alpha,\widehat{S})$ (cf. Section \ref{pre and not} for the definition of the differential).
\end{definition}

It is easily checked that each smooth structural cocycles $(f,\sigma)$ gives rise to a smooth factor system of the form $(\widehat{S}\circ\sigma,f)$. In the following we write $\widehat{N}\times_{(f,\sigma)} G/N$ for the corresponding Lie group extension. 


\begin{proposition}\label{smooth structural cocycle vs. continuation}
Let $(\alpha,\widehat{S})$ be a smooth crossed module. Then the following assertions hold:
\begin{itemize}
\item[(a)]
Each smooth structural cocycle $(f,\sigma)$ gives rise to an extension $$\widehat{G}:=(\widehat{N}\times_{(f,\sigma)} G/N,\widehat{\alpha}_{(f,\sigma)})$$ of $(\alpha,\widehat{S})$ with $\widehat{\alpha}_{(f,\sigma)}:\widehat{N}\times_{(f,\sigma)} G/N\rightarrow G$, $(n,g)\mapsto\alpha(n)\sigma(g)$.
\item[(b)]
Each extension $(\widehat{G},\widehat{\alpha})$ of $(\alpha,\widehat{S})$ gives rise to a smooth structural cocycle by choosing a locally smooth normalized section $\mu:G/N\rightarrow\widehat{G}$ and defining $(f,\sigma):=(\delta_{\mu},\widehat{\alpha}\circ\mu)$. In particular, the extension $(\widehat{G},\widehat{\alpha})$ is equivalent to $$(\widehat{N}\times_{(f,\sigma)} G/N,\widehat{\alpha}_{(f,\sigma)}).$$
\end{itemize}
\end{proposition}
\begin{proof}
(a) Let $(f,\sigma)$ be a smooth structural cocycle and $\mu:G/N\rightarrow\widehat{G}$, $g\mapsto(1_{\widehat{N}},g)$ the canonical section of the corresponding Lie group extension. Then it is easily checked that $(C_{\widehat{G}}\circ\mu,\delta_{\mu})=(\widehat{S}\circ\sigma,f)$. Moreover, a few moments thought shows that 
\begin{align*}
\alpha\circ\delta_{\mu}=\alpha\circ f=\delta_{\sigma} \quad \text{and} \quad \alpha\circ\delta_{\mu}=\alpha\circ f=\delta_{\sigma}.
\end{align*}
 Hence, it follows from Lemma \ref{technical lemma 2} that $\widehat{\alpha}_{(f,\sigma)}$ is the unique Lie group morphism satisfying $\widehat{\alpha}_{(f,\sigma)\mid \widehat{N}}=\alpha$ and $\widehat{\alpha}_{(f,\sigma)}\circ\mu=\sigma $ which in turn shows that 
 \begin{align*}
\widehat{S}\circ\widehat{\alpha}_{(f,\sigma)}\circ\mu=\widehat{S}\circ\sigma=C_{\widehat{G}}\circ\mu.
 \end{align*}
In particular, we conclude that the maps $\widehat{S}\circ\widehat{\alpha}_{(f,\sigma)}$ and $C_{\widehat{G}}$ coincide on $\mu(G/N)$. But on $\widehat{N}$ we also have $\widehat{S}\circ\widehat{\alpha}_{(f,\sigma)}=\widehat{S}\circ\alpha=C_{\widehat{G}}$. Therefore, we conclude that $\widehat{S}\circ\widehat{\alpha}_{(f,\sigma)}=C_{\widehat{G}}$ holds on $\widehat{G}$.

(b) Since $\mu$ is a locally smooth section of $q\circ\widehat{\alpha}$, it follows that $\sigma:=\widehat{\alpha}\circ\mu$ is a locally smooth section of $q$. Therefore $\alpha\circ f=\widehat{\alpha}\circ f=\delta_{\sigma}$ and, as Lie groups, $\widehat{G}$ is equivalent to $\widehat{N}\times_{(C_{\widehat{G}}\circ\mu,f)}G/N$. From this we conclude that $$d_{(C_{\widehat{G}}\circ\mu)}f=1_{\widehat{N}},$$ because $(C_{\widehat{G}}\circ\mu,f)$ is a smooth factor system for $(\widehat{N},G/N)$ (cf. Section \ref{pre and not}). Finally, the equation $C_{\widehat{G}}\circ\mu=\widehat{S}\circ\widehat{\alpha}\circ\mu=\widehat{S}\circ\sigma$ shows that $d_{\widehat{S}\circ\sigma}f=1_{\widehat{N}}$.
\end{proof}

The following result shows that equivalence classes of extensions of the smooth crossed module $(\alpha,\widehat{S})$ correspond to orbits of the group $C^1_s(G/N,\widehat{N})$ in the set of smooth structural cocycles:

\pagebreak[3]
\begin{proposition}\label{orbit of structural cocycles}
For two smooth structural cocycles $(f,\sigma)$, $(f',\sigma')\in Z^2_{ss}(\alpha,\widehat{S})$ the following statements are equivalent:
\begin{itemize}
\item[(a)]
$\widehat{N}\times_{(f,\sigma)}G/N$ and $\widehat{N}\times_{(f',\sigma')}G/N$ are equivalent extensions of $(\alpha,\widehat{S})$.
\item[(b)]
There exists an element $c\in C^1_s(G/N,\widehat{N})$ with 
\[c.(f,\sigma):=(c\ast_{\widehat{S}\circ\sigma}f,(\alpha\circ c)\cdot\sigma)=(f',\sigma'),  
\]where
\[(c\ast_{\widehat{S}\circ\sigma}f)(g,g'):=c(g)(\widehat{S}\circ\sigma)(g)(c(g'))f(g,g')c(gg')^{-1}.
\]
\end{itemize}
If these conditions are satisfied, then the map
\[\varphi:\widehat{N}\times_{(f',\sigma')}G/N\rightarrow \widehat{N}\times_{(f,\sigma)}G/N,\quad (n,g)\mapsto (nc(g),g)
\]is an equivalence of extensions of the smooth crossed module $(\alpha,\widehat{S})$ and all equivalences of extensions $\widehat{N}\times_{(f',\sigma')}G/N\rightarrow \widehat{N}\times_{(f,\sigma)}G/N$ are of this form.
\end{proposition}
\begin{proof}
If $\widehat{N}\times_{(f,\sigma)}G/N$ and $\widehat{N}\times_{(f',\sigma')}G/N$ are equivalent extensions of $(\alpha,\widehat{S})$, then they are also equivalent as Lie group extensions (cf. Remark \ref{remark on equivalence}). Therefore \cite[Proposition II.10]{Ne06} implies that there exists an element $c\in C^1_s(G/N,\widehat{N})$ satisfying $f'=c.f=c\ast_{\widehat{S}\circ\sigma}f$. Moreover, if $\mu$ and $\mu'$ denote the corresponding canonical sections, then
\[\sigma'\cdot\sigma^{-1}=\widehat{\alpha}_{(f,\sigma)}\circ((\varphi\circ\mu')\cdot\mu^{-1})=\alpha\circ h.
\]
For the other direction we only point out that $\widehat{\alpha}_{(f',\sigma')}=\widehat{\alpha}_{(f,\sigma)}\circ\varphi$ follows from the assumption $\sigma'=(\alpha\circ c)\cdot\sigma$
\end{proof}

\begin{theorem}\label{ext} 
Let $\Ext(\alpha,\widehat{S})$ denote the set of all equivalence classes of extensions of the smooth crossed module $(\alpha,\widehat{S})$. Then the following map is a bijection:
\[Z^2_{ss}(\alpha,\widehat{S})/C^1_s(\widehat{N},G/N)\rightarrow\Ext(\alpha,\widehat{S}), \quad [(f,\sigma)]\mapsto[\widehat{N}\times_{(f,\sigma)}G/N].
\]
\end{theorem}
\begin{proof}
The claim directly follows from Proposition \ref{smooth structural cocycle vs. continuation} and Proposition \ref{orbit of structural cocycles}.
\end{proof}

\pagebreak[3]
The next statement gives a useful criteria for the extendability of smooth crossed module:

\begin{corollary}\label{extendability crossed module}\emph{(}\cite[Theorem III.8]{Ne06}\emph{)}.
A crossed module $(\alpha,\widehat{S})$ is extendable if and only if its characteristic class $\chi_{ss}(\alpha,\widehat{S})\in \emph{H}^3_{ss}(G/N,Z)_T$ vanishes.
\end{corollary}
\begin{proof} We first note that Theorem \ref{ext} implies that $(\alpha,\widehat{S})$ is extendable if and only if it admits a smooth structural cocycle. Moreover, the definition of the characteristic class
shows that $(\alpha,\widehat{S})$ admits a smooth structural cocycle if and only if the class $\chi_{ss}(\alpha,\widehat{S}):=[(d_{\widehat{S}\circ\sigma}f)]\in \text{H}^3_{ss}(G/N,Z)_T$ vanishes.
\end{proof}

We are finally ready to state and prove the main result of this appendix:

\begin{theorem}\label{H^2(G,Z)=ext}
Let $(\alpha,\widehat{S})$ be a smooth crossed module with $\Ext(\alpha,\widehat{S})\neq\emptyset$. 
Further, let $(\widehat{N}\times_{(f,\sigma)}G/N,\widehat{\alpha}_{(f,\sigma)})$ be a fixed extension of $(\alpha,\widehat{S})$ represented by the smooth structural cocycle $(f,\sigma)$ \emph{(} cf. Theorem \ref{ext}\emph{)}. Then the following assertions hold:
\begin{itemize}
\item[(a)]
Any other extension $(\widehat{N}\times_{(f',\sigma')}G/N,\widehat{\alpha}_{(f',\sigma')})$ representing an element of the set $\Ext(\alpha,\widehat{S})$ is equivalent to one of the form $(\widehat{N}\times_{(f'',\sigma)} G/N,\widehat{\alpha}_{(f'',\sigma)})$ with
\[f''\cdot f^{-1}\in Z^2_{ss}(G/N,Z)_T
\]
\item[(b)]
The extensions $(\widehat{N}\times_{(f,\sigma)} G/N,\widehat{\alpha}_{(f,\sigma)})$ and $(\widehat{N}\times_{(f',\sigma)} G/N,\widehat{\alpha}_{(f',\sigma)})$ define equivalent extensions if and only if
\[f'\cdot f^{-1}\in B^2_{s}(G/N,Z)_T.
\]
\end{itemize}
\end{theorem}
\begin{proof}
We first note that $\im(\sigma\cdot(\sigma')^{-1})\subseteq N$. Therefore, we may an element $c\in C^1_s(G/N,\widehat{N})$ satisfying $\sigma=(\alpha\circ c)\cdot\sigma'$. If we now define $f'':=c\ast_{\widehat{S}\circ\sigma}f'$, then Proposition \ref{orbit of structural cocycles} implies that the extensions $(\widehat{N}\times_{(f',\sigma')} G/N,\widehat{\alpha}_{(f',\sigma')})$  and $(\widehat{N}\times_{(f'',\sigma)} G/N,\widehat{\alpha}_{(f'',\sigma)})$ are equivalent. Moreover, if $\beta:=f^{-1}\cdot f''$, then a few moments thought shows that $\alpha\circ\beta=(\delta_{\sigma})^{-1}\cdot\delta_{\sigma}=1_N$ and consequently that $\im(\beta)\subseteq Z$. Hence, 
\[1_{\widehat{N}}=d_Tf''=d_T(f\cdot\beta)=d_Tf\cdot d_T\beta=d_T\beta
\]from which we conclude that $\beta\in Z^2_{ss}(G/N,Z)_T$. 

(b) According to Proposition \ref{orbit of structural cocycles}, the extensions $(\widehat{N}\times_{(f,\sigma)} G/N,\widehat{\alpha}_{(f,\sigma)})$ and $(\widehat{N}\times_{(f',\sigma)} G/N,\widehat{\alpha}_{(f',\sigma)})$ are equivalent if and only if there exists $c\in C^1_s(G/N,\widehat{N})$ satisfying $f'=c\ast_{\widehat{S}\circ\sigma}f$ and $\sigma=(\alpha\circ c)\cdot\sigma$. From this we immediately conclude that $\im(c)\subseteq Z$ and therefore that $f'=c\ast_{\widehat{S}\circ\sigma}f=f\cdot d_Tc$. 
\end{proof}


\begin{corollary}\label{H^2(G,Z)}
Let $(\alpha,\widehat{S})$ be a smooth crossed module with $\Ext(\alpha,\widehat{S})\neq\emptyset$. Then the map
\[\emph{H}^2_{ss}(G/N,Z)_T\times\Ext(\alpha,\widehat{S})\rightarrow\Ext(\alpha,\widehat{S})
\]
\[([\beta],[(\widehat{N}\times_{(f,\sigma)} G/N,\widehat{\alpha}_{(f,\sigma)})])\mapsto [(\widehat{N}\times_{(f\cdot\beta,\sigma)} G/N,\widehat{\alpha}_{(f\cdot\beta,\sigma)})]
\]is a well-defined simply transitive action.
\end{corollary}

\end{document}